\documentclass[accepted]{uai2026}%

\usepackage[american]{babel}
\usepackage{algorithm}
\usepackage{algorithmic}
\usepackage{amssymb} 
\usepackage{amsmath}
\usepackage[dvipsnames]{xcolor}
\usepackage{makecell}
\usepackage{titlesec}
\usepackage{amsthm}
\usepackage{multirow}
\usepackage{enumitem}
\usepackage{tikz}
\usetikzlibrary{
  arrows.meta,%
  positioning,%
  decorations.pathreplacing,%
  fit,%
  backgrounds%
}
\usepackage{makecell}

\usepackage[skip=5pt]{caption} %afteskip=0pt, 
\usepackage{subcaption}
\theoremstyle{remark}%
\newtheorem*{remark}{Remark}%

\titlespacing*{\paragraph}%
{0pt}{%
1pt}{%
1em}%

\titlespacing*{\section}%
{0pt}{%
2pt}{%
1pt}%

\titlespacing*{\subsection}
{0pt}{%
4pt}{%
2pt}%

\usepackage{natbib}%
\usepackage{mathtools}%
\usepackage{booktabs}%
\usepackage{tikz}%

\usepackage{xcolor}
\usetikzlibrary{arrows.meta, fit, positioning, shapes.geometric,
                backgrounds, patterns, decorations.pathreplacing}

\usepackage{amsthm}
\newtheorem{theorem}{Theorem}
\newtheorem{lemma}[theorem]{Lemma}

\title{Curvature-Weighted Capacity Allocation: A Minimum Description Length Framework for Layer-Adaptive Large Language Model Optimization}

\author[1]{\href{mailto:ikennaamaefuna@usf.edu}{Theophilus Amaefuna}\thanks{Equal contribution.}}
\author[1]{\href{mailto:hvaidya@usf.edu}{Hitesh Vaidya}\textsuperscript{*}}

\author[1]{\href{mailto:anshumanc@usf.edu}{Anshuman Chhabra}}%
\author[1]{\href{mailto:ankurarjunmali@usf.edu}{Ankur Mali}}%

\affil[1]{%
  Bellini College of Artificial Intelligence, Cybersecurity and Computing\par
  University of South Florida, Tampa, Florida, USA
}
  
\begin{document}
\maketitle

\begin{abstract}
\vspace{-0.5cm}
Layer-wise capacity in large language models is highly non-uniform:
some layers contribute disproportionately to loss reduction, whereas
others are nearly redundant. Existing layer-scoring methods provide
sensitivity estimates but do not give a principled rule for converting
those estimates into allocation or pruning decisions under a global
hardware budget. We introduce a curvature-aware, MDL-inspired framework
built around the layer gain
$\zeta_k^2=g_k^\top\widetilde H_{kk}^{-1}g_k$.
This quantity equals twice the maximal decrease predicted by the
regularized layer-restricted quadratic model and incorporates inverse
local curvature; it is therefore a local surrogate for reducible risk,
not a universal dominance claim over gradient-norm scores.
After normalizing the gains into scores $q_k$, we formulate two convex
programs: one allocates expert slots under diminishing
returns, and the other assigns layer-wise pruning ratios while protecting
high-score layers. Both continuous programs have unique globally optimal
solutions characterized by one dual variable and computable in
$O(K\log(1/\varepsilon))$ time by bisection. We also prove a quadratic
transfer-regret bound: when source and target score vectors differ by at
most $\delta$, the target surrogate cost of the transferred decision is
within $O(\delta^2)$ of the target optimum. Experiments on Mistral-7B and
Gemma-7B show clear allocation gains in some settings and competitive,
though mixed, pruning performance. The framework therefore replaces an
empirical score-to-decision heuristic with a budget-feasible optimization
procedure whose guarantees apply to the stated continuous surrogates.

\noindent Code is available on github repo - \href{https://github.com/TKAI-LAB-Mali/Curvature-Weighted-Capacity-Allocation.git}{TKAI-LAB-Mali/Curvature-Weighted-Capacity-Allocation}
\end{abstract}
\section{Introduction}\label{sec:intro}
\vspace{-1mm}
The representational capacity of a neural network is not uniformly distributed across its layers.%
Empirical studies of large language models (LLMs) consistently reveal
that layers differ substantially in their contribution to the training
objective: some layers hold the bulk of the model's expressive power,
while others are near-redundant, contributing little to no loss reduction
\citep{inequalityOfLayers, vitel2025first}.
This non-uniformity has two practical consequences that current
practice addresses only in isolation.
On one hand, \textit{capacity bottlenecks}: layers where representational
power is insufficient, limit model performance even when global
parameter counts are large.
On the other hand, \textit{capacity redundancy}: layers where parameters
contribute negligibly to the objective, inflates model complexity
without commensurate benefit.

As models scale into the hundreds of billions of parameters, both problems are
exacerbated by the underlying physical reality: hardware constraints impose
strict limits on memory, compute, and communication bandwidth.
The central challenge, therefore, is not simply to make models larger
or smaller, but to \emph{allocate capacity where it matters and remove
it where it does not}, within a global resource budget.

\vspace{-2mm}
\paragraph{The missing ingredient: curvature.}
Existing approaches to layer importance estimation rely primarily on
gradient magnitudes, activation statistics, or held-out accuracy drops
\citep{layerIF, AVSS, layerImportanceEstimation}.
These signals share a common limitation: they do not account for the
local curvature of the loss landscape.
A layer may exhibit a large gradient norm yet reside in a region of
high curvature, where the actual achievable loss reduction per unit of
capacity is small.
Conversely, a layer with a moderate gradient in a flat curvature region
may offer substantial reducible risk.
Without curvature information, capacity decisions are systematically
misallocated.

\paragraph{This work.}
We develop a unified, curvature-aware framework for simultaneously
allocating and pruning model capacity across layers under a global
resource constraint.
Our central quantity is the \emph{curvature-adjusted layer gain}:
\[
  \zeta_k^{2}
  \;=\;
  g_k^{\top}\widetilde{H}_{kk}^{-1}g_k,
\]
where $g_k$ is the layer-$k$ gradient and $\widetilde{H}_{kk}$ is a
positive-definite surrogate for the layer-restricted Hessian block.
We show that $\zeta_k^{2}/2$ is the maximal decrease
of the regularized layer-restricted quadratic model. Appendix~\ref{sec:bias}
bounds its approximation error for the nonlinear empirical objective;
for layers $k$ and $\ell$, the induced ordering is certified whenever
$\tfrac12|\zeta_k^2-\zeta_\ell^2|>\varepsilon_k+\varepsilon_\ell$, where
$\varepsilon_k$ and $\varepsilon_\ell$ bound their respective approximation errors.
With the normalized scores $q_k = \zeta_k^2 / \sum_j \zeta_j^2$, we then derive two
complementary convex programs:
\vspace{-1mm}
\begin{itemize}[noitemsep]
  \item \textbf{Capacity allocation} (Section~\ref{sec:alloc}):
    given a global hardware budget $B$, distribute additional capacity
    (e.g. mixture-of-experts slots) preferentially to
    high-$q_k$ layers, with diminishing log-returns penalizing
    over-allocation.
    The program admits a closed-form \emph{curvature-weighted
    water-filling} solution, computed in $O(K\log 1/\varepsilon)$
    via bisection.

  \item \textbf{Capacity pruning} (Section~\ref{sec:mdl_programs}):
    given a global sparsity target $S$, remove parameters
    aggressively from the low-$q_k$ layers while protecting the high-gain
    layers from degradation.
    The program is strongly convex with a unique closed-form minimizer,
    again computed by bisection.
\end{itemize}
\vspace{-2mm}
\noindent
We further analyze \textbf{transfer stability}
(Appendix~\ref{sec:transfer}): when curvature scores drift between a
source domain and a target domain by $\|q^{(A)} - q^{(B)}\|_2 \le \delta$,
the excess cost of using source-derived allocations on the target task
is bounded by $O(\delta^2)$, with explicit constants given by the
condition number of the target program.
This justifies warm-starting allocation and pruning decisions from
source-domain curvature estimates, a practically important property
for fine-tuning and domain adaptation.
\vspace{-2mm}
\paragraph{Connections to information theory.}
Our objective functions are motivated by \textit{Minimum Description Length}
(MDL) \citep{Rissanen1978Modeling, Rissanen1989StochasticCI}: model complexity is penalized by description length, while
data fit is rewarded by a concave utility reflecting diminishing returns
in code-length reduction \citep{Lotfi2023NonVacuousGB}.
This perspective connects our framework to compression-based
generalization bounds, where shorter valid descriptions can yield tighter
generalization bounds \citep{wilson2025position, SCHMIDHUBER1997857}.
The concrete objectives below are MDL-inspired surrogates rather than
literal prefix-code lengths. Unless the empirical loss is a negative
log-likelihood and the logarithm base and sample-size conversion are specified,
$\zeta_k^2/2$ is measured in units of average loss, not bits.

\paragraph{Contributions.}
We summarize our contributions as follows.
\vspace{-2mm}
\begin{enumerate}[noitemsep]
  \item \textbf{Curvature-adjusted layer gain.}
    We derive $\zeta_k^{2}$ from first principles as twice the
    maximal second-order objective decrease attributable to layer $k$,
    and characterize the approximation error introduced by Tikhonov
    regularization of the Hessian block
    (Lemma~\ref{lem:layeropt}, Appendix~\ref{sec:bias}).

  \item \textbf{Curvature-weighted water-filling.}
    We formulate and solve in closed form a convex allocation program
    that distributes capacity according to $q_k$ under diminishing
    returns and a global hardware budget
    (Theorem~\ref{thm:alloc}).

  \item \textbf{Curvature-protected pruning.}
    We formulate and solve in closed form a strongly convex pruning
    program that concentrates sparsity on low-gain layers while
    meeting a global sparsity target
    (Theorem~\ref{thm:prune}).

  \item \textbf{Transfer stability.}
    We prove an $O(\delta^2)$ transfer regret bound under score drift,
    with explicit constants tied to the condition number and
    score-gradient Lipschitz constant of the target program
    (Theorem~\ref{thm:transfer}).

  \item \textbf{Efficient algorithms.}
    We provide $O(K\log 1/\varepsilon)$ bisection algorithms for both
    programs, with explicit bisection brackets and compatibility with
    standard Hessian approximations
    (Algorithms~\ref{alg:alloc} and~\ref{alg:prune}).
\end{enumerate}

\section{Background}
\label{sec:background}

\paragraph{Minimum Description Length.}
The Minimum Description Length (MDL) principle \citep{Rissanen1978Modeling}
formalizes the tradeoff between model complexity and data fit:
the best model is the one that minimizes the total codelength required
to describe both the model and the data under the model,
\begin{equation}
  \label{eq:mdl}
  \mathrm{Cl}(D)
  \;=\;
  \underbrace{\mathrm{Cl}(\theta)}_{\text{model complexity}}
  \;+\;
  \underbrace{\mathrm{Cl}(D \mid \theta)}_{\text{data fit}},
\end{equation}
where $\mathrm{Cl}(\theta)$ is the codelength (in bits) required to
describe the parameters $\theta$, and $\mathrm{Cl}(D\mid\theta)$ is
the codelength required to describe the data given the model.
Under this principle, a model with lower $\mathrm{Cl}(D)$ simultaneously
achieves better compression and generalization, as well as lower unnecessary
complexity \citep{Lotfi2023NonVacuousGB, Prada2025BridgingPC}.

\paragraph{MDL for capacity allocation.}
Adding capacity at layer $k$ (e.g. adding
mixture-of-experts slots by $e_k \ge 0$) increases model codelength by
$\mathrm{Cl}(\theta_{e_k})$ while reducing data-fit codelength by
$\Delta_{e_k}\mathrm{Cl}(D\mid\theta)$ \citep{pac-mdl}.
Since the base terms $\mathrm{Cl}(\theta)$ and $\mathrm{Cl}(D\mid\theta)$
are constant with respect to the allocation decision, minimizing the total
codelength Eq.~\eqref{eq:mdl} over $e = (e_1,\ldots,e_K)$ reduces to,
\begin{equation}
  \label{eq:mdl_alloc}
  \min_{e_k \ge 0}
  \sum_{k=1}^{K}
  \Bigl[
    \mathrm{Cl}(\theta_{e_k})
    -
    \Delta_{e_k}\mathrm{Cl}(D \mid \theta)
  \Bigr].
\end{equation}
The first term penalizes complexity growth; the second rewards data-fit
improvement. We instantiate Eq.~\eqref{eq:mdl_alloc} concretely in
Section~\ref{sec:alloc} by modeling $\mathrm{Cl}(\theta_{e_k}) \propto \alpha c_k e_k$
(linear in resource usage) and
$\Delta_{e_k}\mathrm{Cl}(D\mid\theta) \propto \gamma q_k^\beta \log(1+e_k)$
(concave, reflecting diminishing returns in code-length reduction).These proportional models specify a tractable MDL-inspired objective; they are not derived as exact codelength identities. Consequently, the optimization results below establish optimality for the stated surrogate, not universal MDL optimality over all possible codes.

\paragraph{MDL for pruning.}
Pruning layer $k$ by sparsity ratio $\rho_k \in [0,1]$ reduces the model
codelength by $-\Delta_{\rho_{k}}\mathrm{Cl}(\theta)$ but increases data-fit
codelength by $\Delta_{\rho_{k}}\mathrm{Cl}(D\mid\theta)$ \citep{lotteryTicketHypothesis}.
Minimizing total codelength over $\rho_k = (\rho_1,\ldots,\rho_K)$
subject to a global sparsity target $S$%
\begin{equation}
  \label{eq:mdl_prune}
      \min_{0 \le \rho_k \le 1}
  \sum_{k=1}^{K}
  \Bigl[
    \Delta_{\rho_k}\mathrm{Cl}(\theta)
    +
    \Delta_{\rho_k}\mathrm{Cl}(D \mid \theta)
  \Bigr].
\end{equation}

We instantiate Eq.~\eqref{eq:mdl_prune} in Section~\ref{sec:mdl_programs}
by modeling $\Delta_{\rho_k}\mathrm{Cl}(\theta) = -b\,n_k\rho_k$
(bits saved by removing $n_k\rho_k$ parameters) and
$\Delta_{\rho_k}\mathrm{Cl}(D\mid\theta) = \eta\,q_k^\kappa\rho_k^2$
(convex degradation penalty, weighted by layer quality $q_k$).
The quadratic degradation term is likewise a modeling surrogate. Its empirical adequacy is separate from the convexity and closed-form analysis of the resulting program.

\paragraph{Layer quality via second-order information.} Both programs Eq.~\eqref{eq:mdl_alloc} and Eq.~\eqref{eq:mdl_prune} depend on
layer quality scores $q_k$ that measure how much each layer contributes
to reducible empirical risk.
A natural candidate is the \emph{Newton decrement} restricted to layer
$k$: given gradient $g_k = E_k\nabla L(\theta)$ and positive-definite
Hessian surrogate $\widetilde{H}_{kk} = E_k\nabla^2 L(\theta)E_k^\top + \tau I$,
the quantity $\zeta_k^{2}$
is twice the maximal second-order decrease in $L$ achievable by
updating layer $k$ alone (derived in Section~\ref{sec:secondorder}).
We use normalized scores:
\(
  q_k
  \;=\;
  \zeta_k^{2}\bigg/\sum_{j=1}^{K}\zeta_j^{2},
\)
throughout, ensuring scale invariance with respect to the global
curvature magnitude.
Layers with large $q_k$ carry more reducible risk and should receive
more capacity; layers with small $q_k$ are candidates for pruning.

\paragraph{Layer Influence scores and their limitations.}
The closest prior measure to $\zeta_k^2$ is the \emph{Layer Influence
score} (LayerIF) of \citet{layerIF}, which localizes the classical
influence function \citep{koh2017understanding} to individual layers.
Given training data $D^{\mathrm{train}} = \{x_i\}_{i=1}^{m}$ and
validation data $D^{\mathrm{valid}} = \{x_j^\nu\}_{j=1}^{n}$,
the global influence score of training sample $x_i$ is
\begin{equation}
  \label{eq:IF}
  I(x_i)
  \;=\;
  -\sum_{j=1}^{n}
  \nabla_{\theta}\ell(x_j^\nu, \theta)^{\top}
  H(\theta)^{-1}
  \nabla_{\theta}\ell(x_i, \theta),
\end{equation}
measuring the sensitivity of validation loss to upweighting $x_i$.
The LayerIF score restricts Eq.~\eqref{eq:IF} to layer $k$ by replacing
full-model quantities with their layer-$k$ counterparts:
\begin{multline}
I^{(k)}(x_i)
\;=\; \\
-\sum_{j=1}^{n}
\nabla_{\theta^{(k)}}\ell(x_j^\nu, \theta)^{\top}
\bigl(H^{(k)}(\theta)\bigr)^{-1}
\cdot\nabla_{\theta^{(k)}}\ell(x_i, \theta).
\label{eq:layerIF}
\end{multline}
A large $|I^{(k)}(x_i)|$ indicates that layer $k$ is sensitive to
the training sample $x_i$, and \citet{layerIF} used the aggregated
magnitude as a proxy for layer quality to guide expert allocation and
pruning.
LayerIF combines training and validation gradients through an
inverse-Hessian operator. Its magnitude therefore reflects gradient alignment
and inverse curvature and cannot, by itself, be identified with high curvature.
The gain $\zeta_k^2/2$ instead has an exact interpretation as the decrease of
the regularized quadratic surrogate. When $L$ is an average negative
log-likelihood measured in nats, the corresponding local total-codelength scale
is $n\zeta_k^2/(2\log 2)$ bits, up to damping and Taylor-remainder errors.

\looseness-1 While LayerIF captures data-dependent sensitivity, it has two
structural limitations that motivate our approach.
\textit{\textbf{First}}, $I^{(k)}$ depends on individual training samples and must be
aggregated heuristically into a layer-level score;
in contrast, $\zeta_k^2$ is defined directly on the empirical
objective and has a closed-form interpretation as reducible risk.
\textit{\textbf{Second}}, and more importantly, LayerIF provides a \emph{signal} but
not an \emph{objective}: translating $\{I^{(k)}\}$ into concrete
expert counts or pruning ratios requires a separate heuristic
(a knapsack assignment with first-come-first-served residual
allocation \citep{layerIF}), with no budget constraint and no
optimality guarantee.
Our MDL-inspired programs Eq.~\eqref{eq:mdl_alloc} and Eq.~\eqref{eq:mdl_prune}
replace this two-stage heuristic with a single convex program
that jointly determines the allocation of all layers under a global
resource constraint, with closed-form solutions and provable
optimality.

We study the problem of allocating limited model capacity across layers
to maximize locally reducible empirical risk under a global resource constraint.
The method consists of two components:
(i)~a curvature-aware measure of layer-wise reducible risk derived from a
second-order expansion of the training objective, and
(ii)~a convex allocation program that distributes capacity according to these
gains under diminishing returns. See Appendix~\ref{sec:related_work} for related works.

\subsection{Objective and Notation}

Let:
\(
  L(\theta)
  = \frac{1}{n}\sum_{i=1}^{n} \varphi\!\left(f(x_i;\theta),\, y_i\right),
\)
denote the empirical objective, where $\varphi$ is a per-sample loss
and $f(\cdot\,;\theta)$ is the model.
We write the gradient and Hessian at $\theta \in \mathbb{R}^{p}$ as:
\(
  g \;:=\; \nabla_{\theta} L(\theta),
  \quad
  H \;:=\; \nabla_{\theta}^{2} L(\theta).
\)
Partition $\theta$ into $K$ disjoint layer blocks
$\theta = (\theta_1, \ldots, \theta_K)$, where $\theta_k \in \mathbb{R}^{p_k}$
and $\sum_{k} p_k = p$.
For each layer $k$, let $E_k \in \{0,1\}^{p_k \times p}$ be the
coordinate-selection matrix and define:
\(  g_k \;:=\; E_k g \in \mathbb{R}^{p_k},
   \quad
  H_{kk} \;:=\; E_k H E_k^{\top} \in \mathbb{R}^{p_k \times p_k}.
\)

\vspace{2mm}
\subsection{Second-Order Expansion and Layer-Restricted Decrease}
\label{sec:secondorder}

\paragraph{Second-order Taylor expansion.}
Suppose $\nabla^{2} L$ is locally Lipschitz with constant $M > 0$
in a neighborhood of $\theta$. 
Taylor's theorem with integral remainder gives, \\
\begin{multline}
  \label{eq:taylor}
  L(\theta + \Delta) - L(\theta)
  =
  g^{\top}\Delta
  + \frac{1}{2}\,\Delta^{\top} H \Delta
  + R(\Delta),
  \\
  |R(\Delta)| \;\le\; \frac{M}{6}\,\|\Delta\|^{3}.
\end{multline}

\paragraph{Layer-restricted quadratic model.}
Restrict perturbations to layer $k$ by setting
$\Delta = E_k^{\top} d_k$ for $d_k \in \mathbb{R}^{p_k}$.
Since $E_k E_k^{\top} = I_{p_k}$ and the blocks are disjoint,
substitution into Eq.~\eqref{eq:taylor} yields the layer-restricted quadratic:
\[
  Q_k(d_k)
  \;=\;
  g_k^{\top} d_k
  + \frac{1}{2}\,d_k^{\top} H_{kk}\, d_k.
\]
Because neural network Hessians are generally indefinite,
minimizing $Q_k$ directly is ill-posed.
We introduce the Tikhonov-regularized surrogate:
\begin{equation}
  \label{eq:surrogate}
  \widetilde{H}_{kk} \;:=\; H_{kk} + \tau I,
  \qquad \tau > 0,
\end{equation}
and analyze
$\widetilde{Q}_k(d_k) = g_k^{\top} d_k + \tfrac{1}{2}\,d_k^{\top}\widetilde{H}_{kk}\,d_k$.
This regularization is standard in second-order deep learning
\citep{martens2015optimizing},\citet{botev2017practical}.
Quadratic damping is the Lagrangian form associated with a
trust-region subproblem for a suitable multiplier; an arbitrary fixed $\tau$
\begin{lemma}[Layer-restricted optimum]
\label{lem:layeropt}
  If $\widetilde{H}_{kk} \succ 0$, the unique minimizer of $\widetilde{Q}_k(d_k)$
  over $d_k \in \mathbb{R}^{p_k}$ is:
  \[
    d_k^{\star} \;=\; -\widetilde{H}_{kk}^{-1}\,g_k,
  \]
  and the corresponding decrease equals,
  \[    \min_{d_k}\,\widetilde{Q}_k(d_k)
    \;=\;
    -\frac{1}{2}\,g_k^{\top}\widetilde{H}_{kk}^{-1}\,g_k.\]
\end{lemma}
The proof for lemma \ref{lem:layeropt} is shown in Appendix \ref{proof:lemma1}.

\paragraph{Curvature-adjusted layer gain.}
Define:%
\begin{equation}
  \label{eq:gain}
  \zeta_k^{2} \;:=\; g_k^{\top}\widetilde{H}_{kk}^{-1}\,g_k \;\ge\; 0.
\end{equation}
By Lemma~\ref{lem:layeropt}, $\zeta_k^2/2$ is exactly the
maximal decrease predicted by the regularized quadratic model $\widetilde Q_k$.
For $d_k^\star=-\widetilde H_{kk}^{-1}g_k$ and
$\Delta_k^\star=E_k^\top d_k^\star$, Taylor's theorem gives,

\begin{multline*}
L(\theta+\Delta^\star)-L(\theta)
= -\frac12\zeta_k^2-\frac{\tau}{2}\|d_k^\star\|_2^2+R(\Delta^\star),
\qquad\\
|R(\Delta^\star)|\le \frac{M}{6}\|d_k^\star\|_2^3.
\end{multline*}
Thus $\zeta_k^2$ is a local surrogate for reducible empirical risk rather
than a global characterization of the nonlinear loss landscape. Moreover,
\[
\frac{\|g_k\|_2^2}{\lambda_{\max}(\widetilde H_{kk})}
\le \zeta_k^2
\le
\frac{\|g_k\|_2^2}{\lambda_{\min}(\widetilde H_{kk})},
\]
so the score incorporates curvature but need not preserve the ranking induced
by $\|g_k\|_2^2$. Appendix~\ref{sec:bias} gives the associated approximation
and ranking conditions.

\subsection{Capacity Allocation Under Diminishing Returns}
\label{sec:alloc}
\paragraph{Setup.}
Let $e_k \ge 0$ denote continuous capacity allocated to layer $k$
(e.g., number of active units or effective precision bits),
with per-unit resource cost $c_k > 0$.
We impose the global budget constraint
\begin{equation}
  \label{eq:Budget}
  \sum_{k=1}^{K} c_k\,e_k \;\le\; B.
\end{equation}
Define the normalized curvature weights:
\begin{equation}
  \label{eq:qk}
  q_k
  \;:=\;
  \frac{\zeta_k^{2}}{\sum_{j=1}^{K}\zeta_j^{2}},
  \qquad q_k \ge 0,\quad \sum_k q_k = 1
\end{equation}
This definition assumes $\sum_j\zeta_j^2>0$. When zero or
near-zero scores are possible, we use the smoothed simplex weights
\[
q_k^{(\varepsilon)}
=
\frac{\zeta_k^2+\varepsilon_q}{\sum_{j=1}^K\zeta_j^2+K\varepsilon_q},
\qquad \varepsilon_q>0,
\]
which satisfy $q_k^{(\varepsilon)}>0$ and provide the positive lower bound
needed for uniform strong-convexity and score-Lipschitz constants.

\looseness-1\textbf{Optimization program.}
We seek an allocation that concentrates capacity where reducible risk is
largest, while penalizing over-allocation via diminishing log-returns.
Setting $\alpha_k = \alpha c_k$ (so that linear cost scales with resource
usage and no layer-dependent free parameter is introduced),
we solve: 
\begin{multline}
  \label{eq:alloc_program}
  \min_{e_k \ge 0}
  \;\sum_{k=1}^{K}
  \Bigl[
    \alpha c_k\,e_k
    \;-\;
    \gamma\, q_k^{\beta}\phi(e_k)
  \Bigr]
  \\
  \text{s.t.}
  \quad
  \sum_{k=1}^{K} c_k\,e_k \;\le\; B,
\end{multline}

where $\alpha, \gamma > 0$ are scalar hyperparameters,
$\beta \ge 0$ controls gain emphasis (see below) and $\phi(e_k) = \log(1 + e_k)$.
The objective is convex:
$\alpha c_k e_k$ is linear, and $-\gamma q_k^{\beta}\log(1+e_k)$ is convex
on $e_k \ge 0$ since $-\log(1+e)$ is convex.
The feasible set is a convex polytope.
Slater's condition \citet{boyd2004convex} holds (any strictly feasible $e_k > 0$ with
$\sum_k c_k e_k < B$ is a Slater point), so strong duality applies.

\paragraph{Closed-form solution.}
Forming the Lagrangian:
\[
  \mathcal{L}(e,\lambda)
  =
  \sum_k
  \Bigl[\alpha c_k e_k - \gamma q_k^{\beta}\log(1+e_k)\Bigr]
  + \lambda\Bigl(\sum_k c_k e_k - B\Bigr),
\]
and imposing stationarity: $\partial\mathcal{L}/\partial e_k = 0$ for $e_k > 0$ gives,
\[
  \alpha c_k + \lambda c_k
  - \frac{\gamma q_k^{\beta}}{1 + e_k}
  = 0.
\]
Solving and incorporating the non-negativity constraint yields the
\emph{curvature-weighted water-filling} allocation,
\begin{equation}
  \label{eq:alloc_closed}
  \boxed{
    e_k(\lambda)
    \;=\;
    \max\!\left\{
      \frac{\gamma\, q_k^{\beta}}{(\alpha + \lambda)\,c_k}
      - 1,\;
      0
    \right\}.
  }
\end{equation}

If the unconstrained optimum ($\lambda = 0$) violates the budget,
the constraint is active and the dual variable $\lambda^{\star} > 0$
satisfies $\sum_k c_k\,e_k(\lambda^{\star}) = B$.
The map $\lambda\mapsto\sum_k c_k e_k(\lambda)$ is continuous
and nonincreasing and is strictly decreasing on every interval containing at
least one active coordinate. Strict convexity gives a unique primal minimizer;
when the active budget equation is nondegenerate, its dual multiplier is unique
and can be computed in $O(K\log(1/\varepsilon))$ time by bisection. A
zero-allocation plateau can admit several dual values representing the same
primal solution.

\paragraph{Role of $\beta$.}
The exponent $\beta$ interpolates between three regimes:
(i)~$\beta = 0$: $q_k^{\beta} = 1$ uniformly, so curvature information is
    discarded and the allocation depends only on costs $c_k$;
(ii)~$\beta = 1$: capacity is allocated proportionally to normalized gains
    $q_k$, recovering a natural baseline;
(iii)~As $\beta$ increases, relative weight is increasingly
placed on the largest scores. With fixed $\gamma$ and normalized $q_k<1$,
however, all $q_k^\beta$ vanish as $\beta\to\infty$, so the allocation may
collapse to zero. Literal concentration on the maximizers requires rescaling
$\gamma$ with $\beta$ or normalizing the powered weights as
$q_k^\beta/\sum_j q_j^\beta$.
In practice, $\beta = 1$ or $\beta = 2$ work well across architectures
(Section~\ref{sec:experiment}).

\begin{figure*}[t]
\centering
\includegraphics[width=0.7\linewidth]{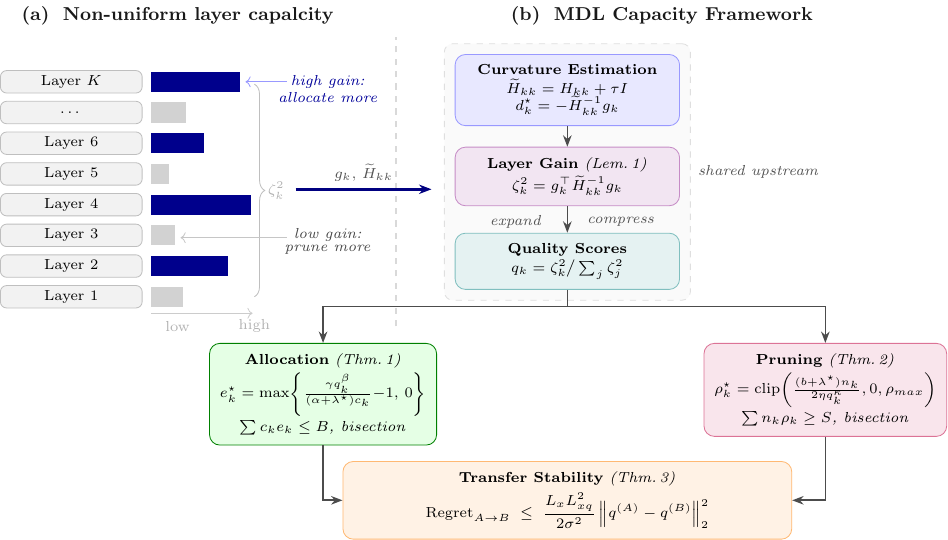}
\captionsetup{font=footnotesize}
\caption{%
  \textbf{(a)}~Layer-wise curvature scores $\zeta_k^2$ vary substantially
  across the transformer stack.
  High-gain layers (dark bars) hold disproportionate reducible risk
  and should receive additional capacity;
  low-gain layers (light bars) are candidates for aggressive pruning.
  \textbf{(b)}~Our framework computes $\zeta_k^2$ from per-layer
  gradients $g_k$ and regularized Hessian blocks $\widetilde{H}_{kk}$,
  normalizes them to quality scores $q_k$,
  and solves two convex programs: a capacity allocation program
  (Theorem~\ref{thm:alloc}) that enriches high-gain layers,
  and a pruning program (Theorem~\ref{thm:prune}) that
  concentrates sparsity on low-gain layers.
  Both programs admit closed-form solutions via $O(K\log 1/\varepsilon)$
  bisection (Algorithms~\ref{alg:alloc}--\ref{alg:prune}).
  Theorem~\ref{thm:transfer} bounds the cost of
  transferring source-domain allocations to a target domain.}
\label{fig:method}
\vspace{-1em}
\end{figure*}

\section{Method}
\label{sec:methods}
In this section, we present our MDL framework for expert allocation and pruning as displayed in Figure \ref{fig:method}.
\subsection{Capacity Allocation and Pruning Under Curvature-Weighted Tradeoffs}
\label{sec:mdl_programs}
The curvature-adjusted layer gains $\zeta_k^{2}$ defined in Eq.~\eqref{eq:gain} and derived in Section~\ref{sec:secondorder}, quantify how much empirical risk can be locally reduced
by updating layer $k$ alone.
We now use these gains to drive two complementary capacity decisions:
allocating additional capacity to under-resourced layers,
and pruning parameters from layers that contribute little to risk reduction (see Figure \ref{fig:method}).
In both programs we work with the normalized quality scores
introduced in Eq.~\eqref{eq:qk}.
Normalization ensures that the allocation parameters $(\gamma,\beta)$
are invariant to the global scale of the curvature signal,
and that $q_k$ can be interpreted directly as the relative share of
reducible risk attributable to layer $k$.
A layer with large $q_k$ has high potential to reduce empirical risk;
accordingly, it should receive more capacity and be protected from pruning.

\textbf{Capacity allocation.}
Given the setup in Section \ref{sec:alloc} and Eq.~\eqref{eq:Budget} let $e_k$ denote the effective capacity added at layer $k$
(e.g., LoRA rank or number of mixture-of-experts slots) \citep{capacityOfNN},
setting $\alpha_k = \alpha c_k$ so that linear penalty scales with
hardware usage (and no layer-dependent free parameter is introduced),
we model the tradeoff between model complexity and risk reduction
via the separable convex program in Eq.~\eqref{eq:alloc_program},
where $\gamma > 0$ scales benefit strength,
$\beta \ge 0$ controls curvature emphasis,
and $\phi$ is increasing and strictly concave with $\phi(0)=0$.
The concavity of $\phi$ models diminishing returns:
each additional unit of capacity yields progressively smaller
reductions in empirical risk.

\begin{theorem}[Convexity and closed form for allocation]
\label{thm:alloc}
The objective in Eq.~\eqref{eq:alloc_program} is convex.
If $\phi$ is strictly concave and $q_k > 0$ for all $k$,
the objective is strictly convex and the minimizer is unique.

For $\phi(e) = \log(1+e)$, there exists a unique $\lambda^{\star} \ge 0$
such that the optimal allocation is as shown in Eq. \eqref{eq:alloc_sol}.
\begin{equation}
  \label{eq:alloc_sol}
  \boxed{
  e_k^{\star}
  \;=\;
  \max\!\left\{
    \frac{\gamma\,q_k^{\beta}}{(\alpha + \lambda^{\star})\,c_k}
    - 1,\;
    0
  \right\}.
  }
\end{equation}
When $\lambda^{\star} > 0$ the budget constraint is active,
$\sum_k c_k e_k^{\star} = B$;
when $\lambda^{\star} = 0$ the unconstrained optimum is feasible.
On the active set $\{k : e_k^{\star} > 0\}$,
$e_k^{\star}$ is strictly increasing in $q_k$. A detailed proof for this theorem is discussed in Appendix~\ref{proof:theorem2}
\end{theorem}

\paragraph{Interpretation.}
At optimality, the marginal benefit equals the marginal cost at every
active layer:
\[
  \underbrace{\gamma q_k^{\beta}\,\phi'(e_k^{\star})}_{\text{marginal benefit}}
  \;=\;
  \underbrace{(\alpha + \lambda^{\star})\,c_k}_{\text{marginal cost}}.
\]
\looseness-1Layers with larger curvature signal $q_k$ receive more capacity,
while the global multiplier $\lambda^{\star}$ enforces the budget.
The exponent $\beta$ interpolates between three regimes:
$\beta = 0$ produces a uniform allocation that ignores curvature;
$\beta = 1$ allocates proportionally to normalized gains;
and as $\beta \to \infty$, allocation concentrates on the layer(s) with
the largest $\zeta_k^{2}$.
In practice, $\beta \in \{1, 2\}$ works well across architectures
(Section~\ref{sec:experiment}).

\textbf{Layer-wise pruning.}
We now consider the complementary problem of removing parameters from
layers that carry little curvature signal.
Let $\rho_k \in [0,1]$ denote the fraction of parameters pruned at
layer $k$, and let $n_k$ denote the total number of parameters in layer $k$.
The number of retained parameters is $n_k(1-\rho_k)$,
contributing $b\,n_k(1-\rho_k)$ bits to model size,
where $b > 0$ is bits per parameter.

Pruning a fraction $\rho_k$ of layer $k$ degrades data fit.
Since $\zeta_k^{2}$ measures the maximal second-order decrease
achievable by updating layer $k$ (Section~\ref{sec:secondorder}),
layers with larger $q_k \propto \zeta_k^{2}$ are more sensitive to
parameter removal.
We model the resulting degradation as $\eta q_k^{\kappa}\psi(\rho_k)$,
where $\psi$ is convex with $\psi(0)=0$,
$\kappa \ge 0$ controls curvature emphasis, and $\eta > 0$ scales the penalty.
This is a design choice grounded in the curvature interpretation of $q_k$;
layers with larger curvature signal attract a higher degradation penalty,
so the optimizer naturally protects them. 

We enforce a minimum global sparsity target $S \le \sum_k n_k \rho_k$
and solve:
\begin{multline}
  \label{eq:prune_program}
  \min_{\rho_1, \ldots, \rho_k}
  \;\sum_{k=1}^{K}
  \Bigl[
    b\,n_k(1-\rho_k)
    \;+\;
    \eta\,q_k^{\kappa}\,\psi(\rho_k)
  \Bigr]
  \\
  \text{s.t.}
  \quad
  \sum_{k=1}^{K} n_k\,\rho_k \;\ge\; S, \,\,\, 0 \le \rho_k \le \rho_{\max}.
\end{multline}
The objective minimizes total model size subject to bounded degradation,
with the constraint enforcing that at least $S$ parameters are pruned globally.

\begin{theorem}[Strong convexity and closed form for pruning]
\label{thm:prune}
The objective in Eq.~\eqref{eq:prune_program} is convex.
If $\psi$ is strongly convex and $q_k^\kappa>0$ for all
$k$, the objective is strongly convex and the primal minimizer is unique. A
uniform modulus requires $q_k\ge q_{\min}>0$, or the smoothed scores above,
because $q_k^\kappa$ vanishes at zero when $\kappa>0$. Feasibility requires
$0\le S\le \rho_{\max}\sum_k n_k$.

For $\psi(\rho) = \rho^{2}$ and $q_k^{\kappa} > 0$ for all $k$,
there exists $\lambda^{\star} \ge 0$ such that
\begin{equation}
  \label{eq:prune_closed}
  \boxed{
  \rho_k^{\star}
  \;=\;
  \operatorname{clip}\!\left(
    \frac{(b + \lambda^{\star})\,n_k}{2\,\eta\,q_k^{\kappa}},
    \;0,\;\rho_{\max}
  \right).
  }
\end{equation}
When $\lambda^{\star} > 0$ the sparsity constraint is active,
$\sum_k n_k\rho_k^{\star} = S$;
when $\lambda^{\star} = 0$, the unconstrained solution already
meets or exceeds the target.
On interior coordinates ($0 < \rho_k^{\star} < \rho_{\max}$),
$\rho_k^{\star}$ is strictly decreasing in $q_k$.  A detailed proof for this theorem is discussed in Appendix~\ref{proof:prune}.
\end{theorem}
\paragraph{Interpretation.}
The pruning solution mirrors the allocation solution but in the opposite
direction: layers with small $q_k$ (low curvature signal, little reducible
risk) are pruned aggressively, while high-gain layers are protected.
The global multiplier $\lambda^{\star}$ calibrates the pruning depth to
meet the sparsity target $S$, playing the same role as the budget multiplier
in the allocation program.
The quadratic choice $\psi(\rho)=\rho^2$ is a surrogate
selected for analytic transparency. Convexity requires convex $\psi$;
monotone degradation is modeled by nondecreasing $\psi$; uniqueness requires
strict or strong convexity; and a closed-form inverse-gradient representation
requires $\psi'$ to be invertible on the relevant interval.
Together, the two programs provide a unified curvature-aware framework:
allocation enriches layers where gradient information is underexploited,
and pruning removes redundancy where unnecessary.
\paragraph{Transfer Stability Under Score Drift}
The allocation and pruning programs depend on the normalized score vector
$q$. A domain change can move this vector from $q^{(A)}$ to $q^{(B)}$.
Appendix~\ref{sec:transfer} bounds the target-objective excess cost incurred
when the source score vector is used in place of the target score vector.

% \vspace{-1mm}

\subsection{Algorithms: Closed-Form Solutions via One-Dimensional Dual Search}
\label{sec:algorithms}
\vspace{-2mm}

The convex programs of Theorems~\ref{thm:alloc} and~\ref{thm:prune}
share a common structure (see Figure~\ref{fig:method}): strict (or strong) convexity guarantees a unique
primal minimizer, and the KKT conditions reduce each constrained program
to a monotone scalar equation in the Lagrange multiplier $\lambda$.
Both programs therefore admit $O(K\log(1/\varepsilon))$ algorithms
via bisection, with primal variables recovered in closed form at each
function evaluation.
Throughout, $q_k = \zeta_k^2 / \sum_j \zeta_j^2$ denotes the normalized
curvature score from Eq.~\eqref{eq:qk}, $\alpha_k = \alpha c_k$ as
established in Section~\ref{sec:alloc}, and we use the canonical utility
choices $\phi(e) = \log(1+e)$ and $\psi(\rho) = \rho^2$. Since capacity is equivalent to an expert in mixture-of-expert, from here on out, we use expert allocation in place of capacity allocation.

\paragraph{Expert Allocation.}

From Theorem~\ref{thm:alloc}, the unique minimizer of Eq.~\eqref{eq:alloc_program}
takes the closed form in Eq.~\eqref{eq:alloc_closed}
where $\lambda \ge 0$ is the dual variable for the budget constraint.
The sum $\lambda \mapsto \sum_k c_k\,e_k(\lambda)$ is continuous and
strictly decreasing: each active term $e_k(\lambda) > 0$ decreases
strictly in $\lambda$, and once $e_k$ hits zero it stays there.
Therefore $\lambda^\star$ is unique when the constraint is active.
To bracket the bisection, note that
\(e_k(0) = \max\{\gamma q_k^\beta / (\alpha c_k) - 1,\,0\}\)
gives the unconstrained solution. A valid upper bracket is

\[
\lambda_{\max}
=
\max\!\left\{0,\;\max_k\left(\frac{\gamma q_k^\beta}{c_k}-\alpha\right)\right\},
\]
because $e_k(\lambda)=0$ whenever
$\lambda\ge \gamma q_k^\beta/c_k-\alpha$.

\begin{figure*}[t]
\centering
\begin{minipage}[t]{0.48\textwidth}
\begin{algorithm}[H]
\caption{MDL-Inspired Expert Allocation}
\label{alg:alloc}
\begin{algorithmic}[1]
\REQUIRE Scores $\{q_k\}$, costs $\{c_k\}$, budget $B$, hyperparameters $\alpha,\gamma,\beta > 0$
\STATE Define closed-form primal: Eq.~\eqref{eq:alloc_closed}
\STATE 
Set $\lambda_{\min} \gets 0$ and
$\lambda_{\max} \gets \max\{0,\max_k(\gamma q_k^\beta/c_k-\alpha)\}$.
\IF{$\sum_k c_k\,e_k(0) \le B$}
    \STATE $\lambda^\star \gets 0$ \COMMENT{budget constraint inactive}
\ELSE
    \STATE Find $\lambda^\star \in (\lambda_{\min}, \lambda_{\max})$ via bisection on $\sum_k c_k\,e_k(\lambda) = B$
\ENDIF
\STATE return $e_k^\star \gets e_k(\lambda^\star)$ for all $k$
\end{algorithmic}
\end{algorithm}
\end{minipage}
\hfill
\begin{minipage}[t]{0.48\textwidth}
\begin{algorithm}[H]
\caption{MDL-Inspired Layer-wise Pruning}
\label{alg:prune}
\begin{algorithmic}[1]
\REQUIRE Scores $\{q_k\}$, sizes $\{n_k\}$, sparsity target $S$, hyperparameters $b,\eta,\kappa > 0$
\STATE Define projected primal: Eq.~\eqref{eq:prune_closed}
\STATE \parbox[t]{\dimexpr\linewidth-2em}{\raggedright
Set $\lambda_{\min}\gets0$ and
$\lambda_{\max}\gets\max\{0,\max_k(2\eta q_k^\kappa\rho_{\max}/n_k-b)\}$.
}\strut
\IF{$\sum_k n_k\,\rho_k(0) \ge S$}
    \STATE $\lambda^\star \gets 0$ \COMMENT{sparsity constraint inactive}
\ELSE
    \STATE Find $\lambda^\star \in (\lambda_{\min}, \lambda_{\max})$ via bisection on $\sum_k n_k\,\rho_k(\lambda) = S$
\ENDIF
\STATE return $\rho_k^\star \gets \rho_k(\lambda^\star)$ for all $k$
\end{algorithmic}
\end{algorithm}
\end{minipage}
\end{figure*}

\paragraph{Layer-wise Pruning.}

From Theorem~\ref{thm:prune}, the unique minimizer of Eq.~\eqref{eq:prune_program} is obtained as follows.
The Lagrangian subtracts the dual term for the lower-bound constraint in Eq.~\eqref{eq:prune_dual},

The full KKT stationarity condition is
\[
-bn_k+\eta q_k^\kappa\psi'(\rho_k)-\lambda n_k-\mu_k+\nu_k=0,
\qquad \lambda,\mu_k,\nu_k\ge0,
\]
where $\mu_k$ and $\nu_k$ correspond to $-\rho_k\le0$ and
$\rho_k-\rho_{\max}\le0$. On an interior coordinate,
$0<\rho_k<\rho_{\max}$, complementary slackness gives
$\mu_k=\nu_k=0$. For $\psi(\rho)=\rho^2$, this yields,

\[
  \rho_k^{\mathrm{int}}(\lambda)
  \;=\;
  \frac{(b+\lambda)\,n_k}{2\,\eta\,q_k^{\kappa}}.
\]
Strong convexity guarantees a unique primal minimizer,
although a plateau of the projected scalar map can make the dual multiplier
non-unique.
The positive numerator $b + \lambda > 0$ (since $b,\lambda \ge 0$)
ensures $\rho_k^{\mathrm{int}} \ge 0$, consistent with the lower-bound
structure of the sparsity constraint.
Projection onto $[0,\rho_{\max}]$ enforces the box constraints Eq.~\eqref{eq:prune_closed}.%
The sum $\lambda \mapsto \sum_k n_k\rho_k(\lambda)$ is continuous and
non-decreasing: the interior solution increases linearly in $\lambda$,
and clipping to $[0,\rho_{\max}]$ preserves monotonicity since each $\rho_k$
is non-decreasing before and after projection.
Strict increase holds whenever at least one coordinate remains in the
interior $(0,\rho_{\max})$, guaranteeing a unique $\lambda^\star$ when binding.

The unconstrained solution at $\lambda = 0$ gives
$\rho_k(0) = \operatorname{clip}(b\,n_k / (2\eta q_k^\kappa),\,0,\,\rho_{\max})$.
Feasibility under the common layer cap requires
$S\le\rho_{\max}\sum_k n_k$. A valid upper bisection bracket is
\[
\lambda_{\max}
=
\max\!\left\{0,\;\max_k\left(\frac{2\eta q_k^\kappa\rho_{\max}}{n_k}-b\right)\right\},
\]
which places every coordinate at its upper cap.

\begin{remark}[Budget flexibility]
The budget $B$ in Algorithm~\ref{alg:alloc} need not equal the
base-model FLOPs used in our experiments.
Any computationally feasible value of $B$ yields a valid allocation that is
globally optimal for the continuous surrogate program;
the bisection in Algorithm~\ref{alg:alloc} adapts automatically
to the chosen constraint.
\end{remark}
\looseness-1\textbf{Complexity and practical notes.}
Each bisection iteration evaluates the primal sum in $O(K)$ time.
Both algorithms therefore run in $O(K\log(1/\varepsilon))$ time to
achieve dual gap $\varepsilon$.
This is substantially cheaper than general-purpose interior-point
methods, which require $O(K^3)$ per iteration.
The curvature scores $\{q_k\}$ enter only through the closed-form
expressions and need not be recomputed during the dual search,
so the algorithms are fully compatible with any approximation of
$\widetilde{H}_{kk}^{-1}g_k$ (e.g., diagonal Fisher, K-FAC, or
randomized Nystr\"{o}m sketches).
Together, Algorithms~\ref{alg:alloc} and~\ref{alg:prune} provide
practical implementations of the curvature-weighted water-filling
framework developed throughout this section.

\vspace{-1mm}
\subsection{Experimental Setup}
\label{sec:experiment}
\vspace{-2mm}

\paragraph{Models.}
All experiments are conducted on two publicly available 7B-parameter
large language models: Mistral-7B-v0.1 \citet{mistral7b} and
Gemma-7B \citet{gemma}.
Parameter-efficient fine-tuning is performed using LoRA-MoE
\citet{MolA}, which augments each layer with a mixture of low-rank
adapters and serves as the capacity expansion mechanism targeted by
Algorithm~\ref{alg:alloc}.

\paragraph{Score estimation.}
The theoretical score is the normalized Newton-decrement gain $q_k$.
For scalability, the experiments instead use the proxy
\[
\widehat q_k=\frac{\widehat s_k}{\sum_j\widehat s_j},
\]
where $\widehat s_k$ is the chosen nonnegative aggregation of LayerIF scores.
Algorithms~\ref{alg:alloc} and~\ref{alg:prune} are instantiated with
$\widehat q$, not with the theoretical Newton-decrement weights $q$.
The convex optimization results apply to any fixed positive weight vector;
the experiments therefore evaluate the allocation and pruning decision rules
under this proxy.

\paragraph{Expert allocation.}
\textit{Datasets.}
We evaluate on five classification and question-answering benchmarks:
CoLA \citet{cola_dataset}, MRPC \citet{mrpc_dataset},
CommonsenseQA \citet{commonsenseqa}, ScienceQA \citet{scienceQ},
and OpenBookQA \citet{openbookQ}.

\textit{Curvature scores.}
Layer-wise influence scores are computed for each (model, dataset)
pair following the procedures of \citet{layerIF} and \citet{datainf}.
We consider two variants of the influence score pool:
\textbf{All}, in which every computed influence score is used,
and \textbf{+ve}, in which only positively influential samples
(those with negative influence score values, indicating a beneficial
effect on validation loss) are retained.
Both variants are evaluated on Mistral-7B; only the \textbf{+ve}
variant is used for Gemma-7B.

\textit{Budget and allocation.}
Following \citet{meo}, the FLOPs of a multi-expert LoRA model are
empirically upper-bounded by the FLOPs of the base network.
We therefore set the budget $B$ in Algorithm~\ref{alg:alloc} equal
to the FLOPs of a standard single-expert layer in the respective
base model, with scaling factors $\sigma = 0.0276$ for Mistral-7B
and $\sigma = 0.02$ for Gemma-7B (i.e., $B = \sigma \times
\text{base model FLOPs}$).
Per-layer expert counts are computed using Algorithm~\ref{alg:alloc},
and fine-tuning follows the protocols of \citet{alphalora} and
\citet{MolA} for 5 epochs per (model, dataset) pair.
The convex program produces continuous allocations $e_k^*$.
We convert them to integer expert counts by setting $m_k=\lfloor e_k^*\rfloor$.
This rule preserves budget feasibility,
$\sum_k c_k m_k\le\sum_k c_k e_k^*\le B$, but it can leave unused budget and
is not, in general, the exact solution of the corresponding integer program.

\paragraph{Layer-wise Pruning.}
\textit{Datasets.}
Calibration uses the C4 dataset \citet{c4_dataset}.
Zero-shot post-pruning evaluation is performed on seven benchmarks:
RTE \citet{rte_task}, OpenBookQA \citet{openbookQ},
ARC-Easy and ARC-Challenge \citet{ARCchallenge},
HellaSwag \citet{hellaswag}, BoolQ \citet{boolq},
and WinoGrande \citet{WinoGrande}.

\textit{Pruning configurations.}
Layer-wise pruning ratios $\rho_k^\star$ are computed using
Algorithm~\ref{alg:prune} with influence scores carried over from
the expert allocation experiments.
The computed ratios are applied under three structural pruning
configurations --- Magnitude \citet{magnitude_prune},
SparseGPT \citet{sparsGPT}, and Wanda \citet{wanda} ---
using the framework of \citet{alphapruning}.
A global sparsity target of $S = 0.5 \times (\text{total parameters})$
is enforced for both models, corresponding to 50\% sparsity.

\textit{Evaluation.}
We report the mean zero-shot accuracy for expert allocation after finetuning and across all seven evaluation benchmarks for each pruning configuration and compare against the baseline \citet{layerIF}.

Refer to Appendix~\ref{appendix:experiments} for hardware and parameter settings used in the experiments.
\vspace{-2mm}
\subsection{Results and Discussion}
\label{subsec:results}
\vspace{-2mm}
\begin{figure}[tbh]
    \centering
    \includegraphics[width=\columnwidth]{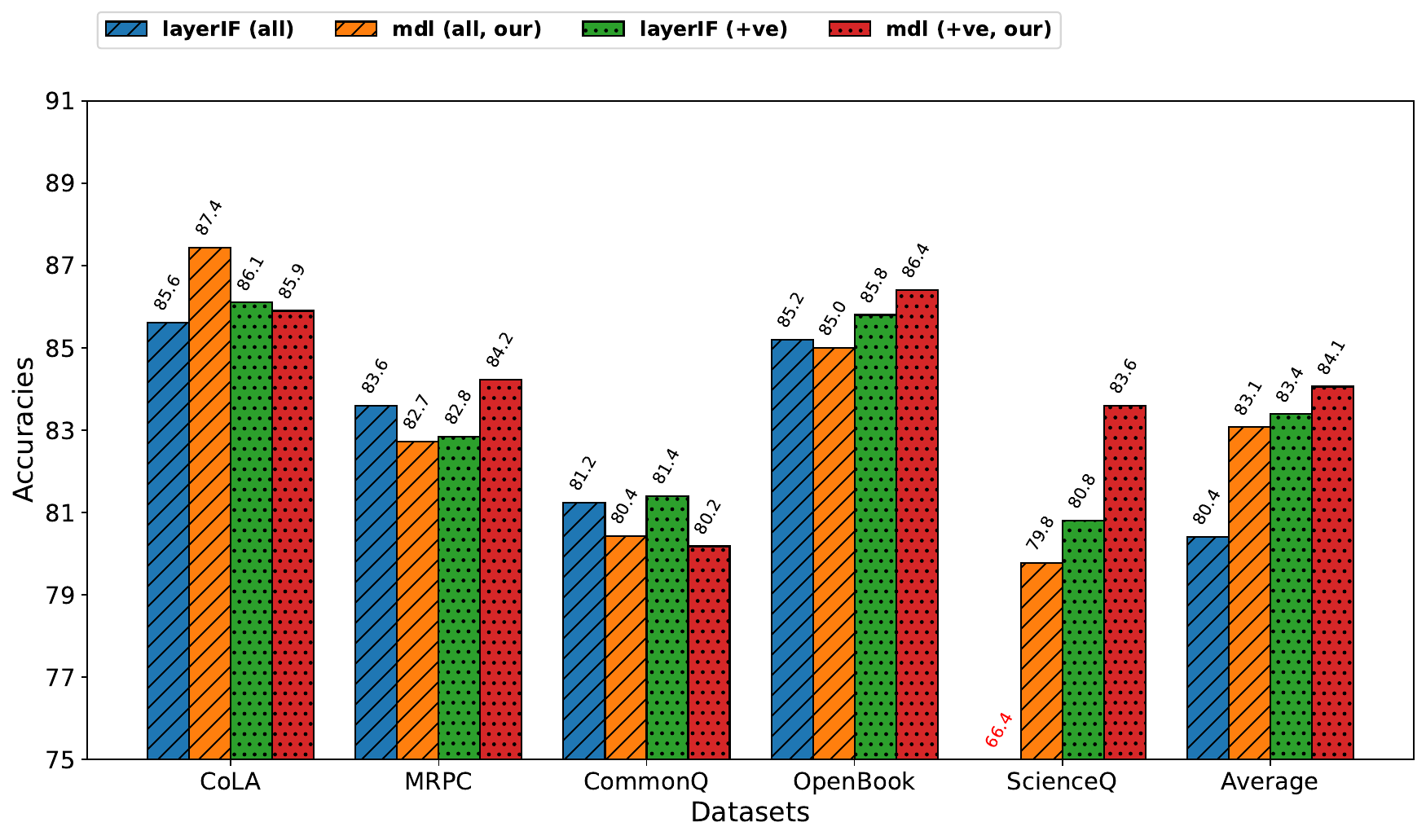}
    \caption{Expert allocation accuracy (\%) on Mistral-7B-v0.1 (5 epochs)}
    \label{fig:mistral_expert_alloc}
\end{figure}
\paragraph{Expert allocation.} Figures~\ref{fig:mistral_expert_alloc} and~\ref{fig:gemma_expert_alloc} report
zero-shot accuracy after five epochs of fine-tuning on Mistral-7B and
Gemma-7B respectively, with per-layer expert counts determined by
Algorithm~\ref{alg:alloc} (MDL) and the LayerIF heuristic baseline.
We note that LayerIF has previously been shown to outperform
AlphaLoRA \citet{alphalora} and MoLA \citet{MolA} on these benchmarks;
we omit those comparisons here for brevity and focus on the
MDL-vs-LayerIF contrast.
On Mistral-7B, the MDL allocation outperforms LayerIF on average
under both influence score variants: \textbf{83.07\%} vs.\ 80.41\%
(All) and \textbf{84.06\%} vs.\ 83.39\% (+ve), corresponding to
absolute improvements of 2.66 and 0.67 percentage points respectively.
The gains are most pronounced on ScienceQA,
where MDL improves over LayerIF by 13.4 points (All) and 2.8 points (+ve),
suggesting that curvature-weighted allocation is particularly beneficial
for knowledge-intensive reasoning tasks where representational capacity
is unevenly demanded across layers.
\begin{figure}[tbh]
    \centering
    \includegraphics[width=\columnwidth]{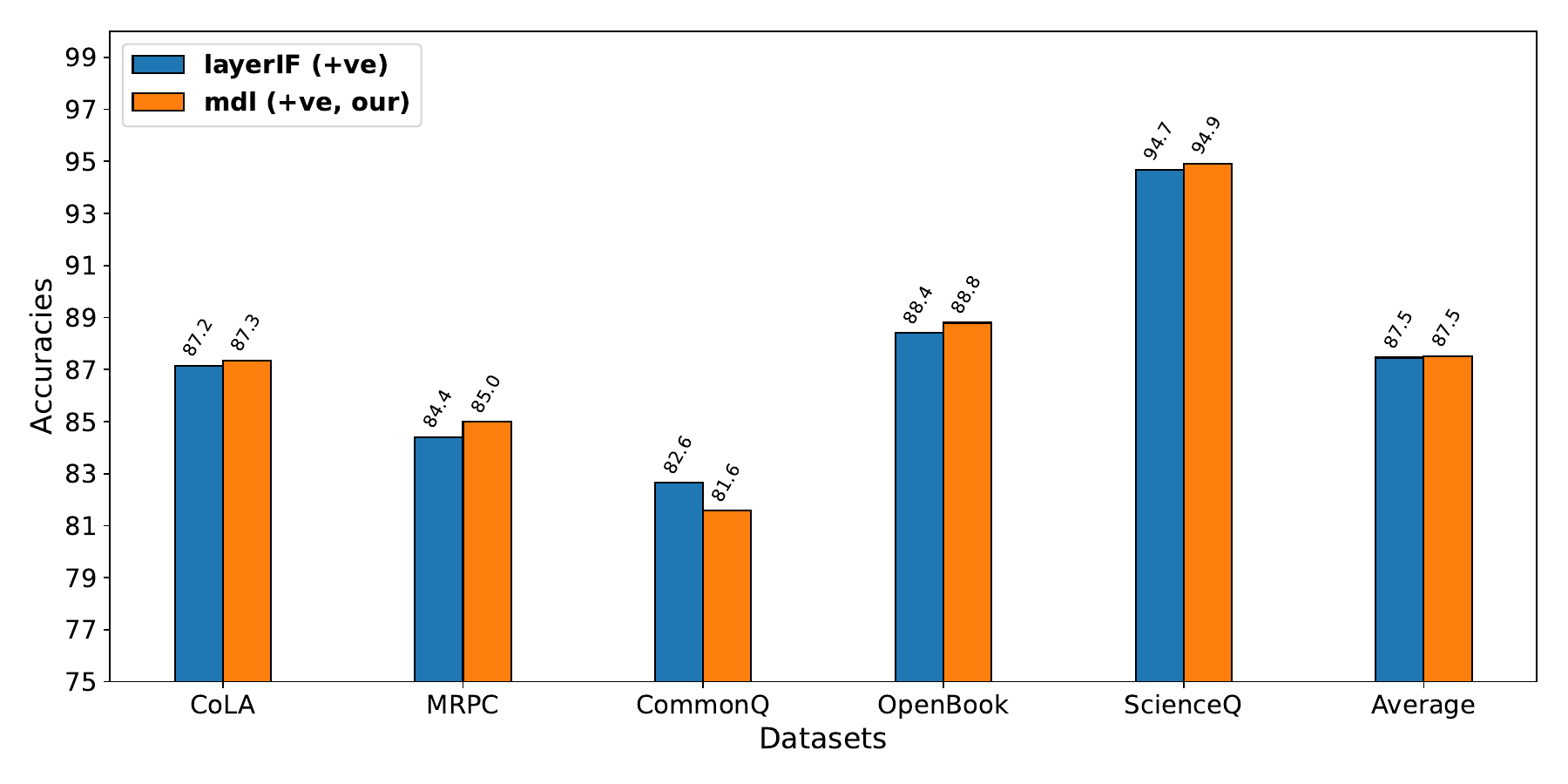}
    \caption{Expert allocation accuracy (\%) on Gemma-7B, +ve variant only (5 epochs)}
    \label{fig:gemma_expert_alloc}
\end{figure}
On Gemma-7B, Algorithm~\ref{alg:alloc} produces identical expert counts
under the \textbf{All} and \textbf{+ve} variant (hence only +ve results are reported).
The MDL allocation yields a marginal improvement to the LayerIF allocation
(\textbf{87.52\%} vs.\ 87.46\%), confirming that the two methods agree
in structure when curvature scores are relatively uniform,
MDL provides a cleaner theoretical justification for the same decision.
These results demonstrate that replacing the knapsack heuristic of
LayerIF with the convex MDL program of Theorem~\ref{thm:alloc}
yields consistent improvements in some cases without additional compute:
Both methods share the same influence-score inputs,
and Algorithm~\ref{alg:alloc} adds only an $O(K\log 1/\varepsilon)$
bisection step. See more in Tables~\ref{tab:mistral_expert_alloc}, \ref{tab:comparison} and~\ref{tab:exp_gemma} for results on Mistral-7B and Gemma-7B respectively. 
\paragraph{Layer-wise Pruning.}
Tables~\ref{tab:pruning_mistral} and~\ref{tab:pruning_gemma} report
mean zero-shot accuracy across seven evaluation benchmarks
at 50\% global sparsity, under Magnitude, Wanda, and SparseGPT
pruning configurations.
We impose the box constraint before optimization, using
$\rho_{\max}=0.51$ for Mistral-7B and $\rho_{\max}=0.55$ for Gemma-7B.
The global 50\% target remains feasible without fully pruning any targeted
parameter block. For Mistral-7B, the combination of a 50\% global target and
$\rho_{\max}=0.51$ leaves little room for heterogeneous layer ratios, which
partly explains the near parity between the two decision rules. Fully pruning
a targeted transformation can cause severe degradation, although residual
connections may still preserve a network-level information path.
\begin{table}[tbh]
\centering
\caption{Mean zero-shot accuracy (\%) across 7 benchmarks for pruning on
         Mistral-7B-v0.1 at 50\% sparsity (layer cap 0.51).
         Rows show the calibration dataset used to compute influence scores.}
\label{tab:pruning_mistral}
\resizebox{\columnwidth}{!}{%
\begin{tabular}{lrrr}
\toprule
\textbf{Calibration dataset} & \textbf{Magnitude} & \textbf{Wanda} & \textbf{SparseGPT} \\
\midrule
CoLA          & 56.47 & 58.59 & 59.66 \\
MRPC          & 56.23 & 58.48 & 60.12 \\
CommonsenseQA & 56.32 & 58.65 & 60.77 \\
OpenBookQA    & 56.22 & 58.46 & 60.05 \\
ScienceQA     & 55.89 & 58.99 & 60.30 \\
\midrule
Average (MDL) & 56.23 & 58.63 & \textbf{60.18} \\
LayerIF       & 56.41 & 58.67 & \textbf{60.18} \\
\bottomrule
\end{tabular}
}
\end{table}
On Mistral-7B (Table~\ref{tab:pruning_mistral}), MDL pruning ratios
match LayerIF closely across all three configurations:
average accuracies are within 0.05 points for Wanda and identical
for SparseGPT (60.18\%), while Magnitude shows a marginal difference
of 0.18 points (56.23\% MDL vs.\ 56.41\% LayerIF).
On Gemma-7B (Table~\ref{tab:pruning_gemma}), MDL outperforms LayerIF
under Magnitude (\textbf{33.34\%} vs.\ 32.91\%) while LayerIF
leads under Wanda (52.3\% vs.\ 49.47\%) and SparseGPT
(50.79\% vs.\ 49.22\%).
The pruning parity between MDL and LayerIF is itself a meaningful
result: it shows that the principled convex program of
Theorem~\ref{thm:prune} recovers the empirically tuned LayerIF ratios
without manual calibration, while providing the theoretical
guarantees --- strong convexity, unique minimizer, and budget feasibility, which are lacking in LayerIF.
The Wanda and SparseGPT gaps on Gemma-7B suggest that the quadratic
degradation model $\psi(\rho) = \rho^2$ may underestimate pruning
sensitivity in certain architectural regimes;
exploring richer $\psi$ (e.g., $\psi(\rho) = -\log(1-\rho)$)
is a natural direction for future work. Overall, the allocation results are strongest on Mistral-7B, while
Gemma-7B allocation is essentially tied with LayerIF and the pruning
results are mixed. The smaller pruning differences are partly explained
by the 50\% global sparsity target and the tight per-layer caps, which
leave limited freedom for heterogeneous pruning ratios. Our optimality
guarantees apply to the continuous surrogate programs; downstream
accuracy can additionally be affected by integer rounding and the
specific structural pruning method.

% \paragraph{Summary.}
% {\color{purple}
% The empirical picture is mixed but informative: the method gives clear
% gains for Mistral-7B allocation, essentially ties LayerIF for Gemma-7B
% allocation, nearly matches it for Mistral-7B pruning, and yields mixed results
% for Gemma-7B pruning. The allocation gains are larger than the pruning gains,
% in part because the 50\% global target and the per-layer caps leave limited
% freedom for heterogeneous pruning ratios. Theoretical optimality refers to the
% continuous surrogate objective, not directly to downstream accuracy after
% integer rounding or structural pruning.
% }

\begin{table}[tbh]
\centering
\caption{Mean zero-shot accuracy (\%) across 7 benchmarks for pruning on Gemma-7B at 50\% sparsity (layer cap 0.55). Rows show the calibration dataset used to compute influence scores.}
\label{tab:pruning_gemma}
\resizebox{\columnwidth}{!}{%
\begin{tabular}{lrrr}
\toprule
\textbf{Calibration dataset} & \textbf{Magnitude} & \textbf{Wanda} & \textbf{SparseGPT} \\
\midrule
CoLA          & 33.03 & 48.36 & 49.28 \\
MRPC          & 33.49 & 52.06 & 50.77 \\
CommonsenseQA & 33.51 & 50.26 & 49.28 \\
OpenBookQA    & 33.52 & 48.09 & 48.11 \\
ScienceQA     & 33.15 & 48.58 & 48.66 \\
\midrule
Average (MDL) & \textbf{33.34} & 49.47 & 49.22 \\
LayerIF       & 32.91          & \textbf{52.30} & \textbf{50.79} \\
\bottomrule
\end{tabular}
}
\end{table}

\section{Conclusion}
\label{sec:conclusion}

We presented a curvature-aware framework for layer-wise capacity
allocation and pruning in large language models, grounded in the
Minimum Description Length principle.
The central quantity, $\zeta_k^2 = g_k^\top\widetilde{H}_{kk}^{-1}g_k$,
measures reducible empirical risk at each layer and drives two convex
programs with unique closed-form solutions, each computable in
$O(K\log 1/\varepsilon)$ via bisection.
Experiments on Mistral-7B and Gemma-7B show clear allocation
gains in one model, an allocation tie in the other, and competitive but mixed
pruning performance. Because both methods use the same influence-derived
proxy scores, the comparison isolates the decision rule. The formal guarantees
are global optimality for the continuous surrogate programs and quadratic
stability with respect to score perturbations under the stated positivity and
strong-convexity assumptions.
See Appendix~\ref{sec:limit} and Appendix~\ref{sec:future} for limitations
and directions for future work, respectively.

\bibliography{uai2026-template}

@inproceedings{
lotteryTicketHypothesis,
title={The Lottery Ticket Hypothesis: Finding Sparse, Trainable Neural Networks},
author={Jonathan Frankle and Michael Carbin},
booktitle={International Conference on Learning Representations},
year={2019},
url={https://openreview.net/forum?id=rJl-b3RcF7},
}

@article{SCHMIDHUBER1997857,
title = {Discovering Neural Nets with Low Kolmogorov Complexity and High Generalization Capability},
journal = {Neural Networks},
volume = {10},
number = {5},
pages = {857-873},
year = {1997},
issn = {0893-6080},
doi = {https://doi.org/10.1016/S0893-6080(96)00127-X},
url = {https://www.sciencedirect.com/science/article/pii/S089360809600127X},
author = {Jürgen Schmidhuber}
}

@article{inequalityOfLayers,
author = {Zhang, Chiyuan and Bengio, Samy and Singer, Yoram},
title = {Are all layers created equal?},
year = {2022},
issue_date = {January 2022},
publisher = {JMLR.org},
volume = {23},
number = {1},
issn = {1532-4435},
journal = {J. Mach. Learn. Res.},
month = jan,
articleno = {67},
numpages = {28}
}

@article{capacityOfNN,
author = {Baldi, Pierre and Vershynin, Roman},
title = {The capacity of feedforward neural networks},
year = {2019},
issue_date = {Aug 2019},
publisher = {Elsevier Science Ltd.},
address = {GBR},
volume = {116},
number = {C},
issn = {0893-6080},
url = {https://doi.org/10.1016/j.neunet.2019.04.009},
doi = {10.1016/j.neunet.2019.04.009},
journal = {Neural Netw.},
month = aug,
pages = {288–311},
numpages = {24}
}

@inproceedings{vitel2025first,
  title={First is {N}ot {R}eally {B}etter {T}han {L}ast: {E}valuating {L}ayer {C}hoice and {A}ggregation {S}trategies in {L}anguage {M}odel {D}ata {I}nfluence {E}stimation},
  author={Vitel, Dmytro and Chhabra, Anshuman},
  booktitle={International {C}onference on {L}earning {R}epresentations},
  year={2026}
}

@inproceedings{botev2017practical,
  title     = {Practical Gauss-Newton Optimisation for Deep Learning},
  author    = {Botev, Aleksandar and Ritter, Hippolyt and Barber, David},
  booktitle = {Proceedings of the 34th International Conference on Machine Learning},
  series    = {Proceedings of Machine Learning Research},
  volume    = {70},
  pages     = {557--565},
  year      = {2017},
  editor    = {Precup, Doina and Teh, Yee Whye},
  publisher = {PMLR},
  address   = {Sydney, Australia},
  url       = {https://proceedings.mlr.press/v70/botev17a.html}
}

@InProceedings{pac-mdl,
author="Blum, Avrim
and Langford, John",
editor="Sch{\"o}lkopf, Bernhard
and Warmuth, Manfred K.",
title="PAC-MDL Bounds",
booktitle="Learning Theory and Kernel Machines",
year="2003",
publisher="Springer Berlin Heidelberg",
address="Berlin, Heidelberg",
pages="344--357",
isbn="978-3-540-45167-9"
}

@article{Prada2025BridgingPC,
  title={Bridging Predictive Coding and MDL: A Two-Part Code Framework for Deep Learning},
  author={Benjamin Prada and Shion Matsumoto and Abdul Malik Zekri and Ankur Mali},
  journal={ArXiv},
  year={2025},
  volume={abs/2505.14635},
  url={https://api.semanticscholar.org/CorpusID:278768469}
}

@inproceedings{Rissanen1989StochasticCI,
  title={Stochastic Complexity in Statistical Inquiry},
  author={Jorma Rissanen},
  booktitle={World Scientific Series in Computer Science},
  year={1989},
  url={https://api.semanticscholar.org/CorpusID:9365056}
}

@inproceedings{Lotfi2023NonVacuousGB,
  title={Non-Vacuous Generalization Bounds for Large Language Models},
  author={Sanae Lotfi and Marc Finzi and Yilun Kuang and Tim G. J. Rudner and Micah Goldblum and Andrew Gordon Wilson},
  booktitle={International Conference on Machine Learning},
  year={2023},
  url={https://api.semanticscholar.org/CorpusID:266573256}
}

@book{ShalevShwartz2014UnderstandingML,
author = {Shalev-Shwartz, Shai and Ben-David, Shai},
address = {Cambridge},
booktitle = {Understanding machine learning : from theory to algorithms},
isbn = {1-107-29801-6},
language = {eng},
publisher = {Cambridge University Press},
title = {Understanding machine learning : from theory to algorithms },
year = {2014},
}

@article{Valiant1984ATO,
  title={A theory of the learnable},
  author={Leslie G. Valiant},
  journal={Commun. ACM},
  year={1984},
  volume={27},
  pages={1134-1142},
  url={https://api.semanticscholar.org/CorpusID:59712}
}

@INPROCEEDINGS{layerImportanceEstimation,
  author={Liu, Hongyang and Elkerdawy, Sara and Ray, Nilanjan and Elhoushi, Mostafa},
  booktitle={2021 IEEE/CVF Conference on Computer Vision and Pattern Recognition Workshops (CVPRW)}, 
  title={Layer Importance Estimation with Imprinting for Neural Network Quantization}, 
  year={2021},
  volume={},
  number={},
  pages={2408-2417},
  doi={10.1109/CVPRW53098.2021.00273}}

@misc{AVSS,
      title={Layer Importance and Hallucination Analysis in Large Language Models via Enhanced Activation Variance-Sparsity}, 
      author={Zichen Song and Sitan Huang and Yuxin Wu and Zhongfeng Kang},
      year={2024},
      eprint={2411.10069},
      archivePrefix={arXiv},
      primaryClass={cs.CL},
      url={https://arxiv.org/abs/2411.10069}, 
}

@misc{mistral7b,
      title={Mistral 7B}, 
      author={Albert Q. Jiang and Alexandre Sablayrolles and Arthur Mensch and Chris Bamford and Devendra Singh Chaplot and Diego de las Casas and Florian Bressand and Gianna Lengyel and Guillaume Lample and Lucile Saulnier and Lélio Renard Lavaud and Marie-Anne Lachaux and Pierre Stock and Teven Le Scao and Thibaut Lavril and Thomas Wang and Timothée Lacroix and William El Sayed},
      year={2023},
      eprint={2310.06825},
      archivePrefix={arXiv},
      primaryClass={cs.CL},
      url={https://arxiv.org/abs/2310.06825}, 
}

@misc{gemma,
      title={Gemma: Open Models Based on Gemini Research and Technology}, 
      author={Gemma Team and Thomas Mesnard and Cassidy Hardin and Robert Dadashi and Surya Bhupatiraju and Shreya Pathak and Laurent Sifre and Morgane Rivière and Mihir Sanjay Kale and Juliette Love and Pouya Tafti and Léonard Hussenot and Pier Giuseppe Sessa and Aakanksha Chowdhery and Adam Roberts and Aditya Barua and Alex Botev and Alex Castro-Ros and Ambrose Slone and Amélie Héliou and Andrea Tacchetti and Anna Bulanova and Antonia Paterson and Beth Tsai and Bobak Shahriari and Charline Le Lan and Christopher A. Choquette-Choo and Clément Crepy and Daniel Cer and Daphne Ippolito and David Reid and Elena Buchatskaya and Eric Ni and Eric Noland and Geng Yan and George Tucker and George-Christian Muraru and Grigory Rozhdestvenskiy and Henryk Michalewski and Ian Tenney and Ivan Grishchenko and Jacob Austin and James Keeling and Jane Labanowski and Jean-Baptiste Lespiau and Jeff Stanway and Jenny Brennan and Jeremy Chen and Johan Ferret and Justin Chiu and Justin Mao-Jones and Katherine Lee and Kathy Yu and Katie Millican and Lars Lowe Sjoesund and Lisa Lee and Lucas Dixon and Machel Reid and Maciej Mikuła and Mateo Wirth and Michael Sharman and Nikolai Chinaev and Nithum Thain and Olivier Bachem and Oscar Chang and Oscar Wahltinez and Paige Bailey and Paul Michel and Petko Yotov and Rahma Chaabouni and Ramona Comanescu and Reena Jana and Rohan Anil and Ross McIlroy and Ruibo Liu and Ryan Mullins and Samuel L Smith and Sebastian Borgeaud and Sertan Girgin and Sholto Douglas and Shree Pandya and Siamak Shakeri and Soham De and Ted Klimenko and Tom Hennigan and Vlad Feinberg and Wojciech Stokowiec and Yu-hui Chen and Zafarali Ahmed and Zhitao Gong and Tris Warkentin and Ludovic Peran and Minh Giang and Clément Farabet and Oriol Vinyals and Jeff Dean and Koray Kavukcuoglu and Demis Hassabis and Zoubin Ghahramani and Douglas Eck and Joelle Barral and Fernando Pereira and Eli Collins and Armand Joulin and Noah Fiedel and Evan Senter and Alek Andreev and Kathleen Kenealy},
      year={2024},
      eprint={2403.08295},
      archivePrefix={arXiv},
      primaryClass={cs.CL},
      url={https://arxiv.org/abs/2403.08295}, 
}

@inproceedings{
wilson2025position,
title={Position: Deep Learning is Not So Mysterious or Different},
author={Andrew Gordon Wilson},
booktitle={Forty-second International Conference on Machine Learning Position Paper Track},
year={2025},
url={https://openreview.net/forum?id=42Au7FoD8F}
}

@inproceedings{meo,
    title = "Merging Experts into One: Improving Computational Efficiency of Mixture of Experts",
    author = "He, Shwai  and
      Fan, Run-Ze  and
      Ding, Liang  and
      Shen, Li  and
      Zhou, Tianyi  and
      Tao, Dacheng",
    editor = "Bouamor, Houda  and
      Pino, Juan  and
      Bali, Kalika",
    booktitle = "Proceedings of the 2023 Conference on Empirical Methods in Natural Language Processing",
    month = dec,
    year = "2023",
    address = "Singapore",
    publisher = "Association for Computational Linguistics",
    url = "https://aclanthology.org/2023.emnlp-main.907/",
    doi = "10.18653/v1/2023.emnlp-main.907",
    pages = "14685--14691"
}

@inproceedings{layerIF,
title={Layer{IF}: {E}stimating {L}ayer {Q}uality for {L}arge {L}anguage {M}odels using {I}nfluence {F}unctions},
author={Hadi Askari and Shivanshu Gupta and Fei Wang and Anshuman Chhabra and Muhao Chen},
booktitle={Advances in {N}eural {I}nformation {P}rocessing {S}ystems},
year={2025}
}

@inproceedings{alphapruning,
title={AlphaPruning: Using Heavy-Tailed Self Regularization Theory for Improved Layer-wise Pruning of Large Language Models},
author={Lu, Haiquan and Zhou, Yefan and Liu, Shiwei and Wang, Zhangyang and Mahoney, Michael W and Yang, Yaoqing},
booktitle={Thirty-eighth Conference on Neural Information Processing Systems},
year={2024}
}

@misc{MolA,
      title={Higher Layers Need More LoRA Experts}, 
      author={Chongyang Gao and Kezhen Chen and Jinmeng Rao and Baochen Sun and Ruibo Liu and Daiyi Peng and Yawen Zhang and Xiaoyuan Guo and Jie Yang and VS Subrahmanian},
      year={2024},
      eprint={2402.08562},
      archivePrefix={arXiv},
      primaryClass={cs.CL},
      url={https://arxiv.org/abs/2402.08562}, 
}

@misc{alphalora,
      title={AlphaLoRA: Assigning LoRA Experts Based on Layer Training Quality}, 
      author={Peijun Qing and Chongyang Gao and Yefan Zhou and Xingjian Diao and Yaoqing Yang and Soroush Vosoughi},
      year={2024},
      eprint={2410.10054},
      archivePrefix={arXiv},
      primaryClass={cs.CL},
      url={https://arxiv.org/abs/2410.10054}, 
}

@misc{sparsGPT,
      title={SparseGPT: Massive Language Models Can Be Accurately Pruned in One-Shot}, 
      author={Elias Frantar and Dan Alistarh},
      year={2023},
      eprint={2301.00774},
      archivePrefix={arXiv},
      primaryClass={cs.LG},
      url={https://arxiv.org/abs/2301.00774}, 
}

@inproceedings{magnitude_prune,
author = {Han, Song and Pool, Jeff and Tran, John and Dally, William J.},
title = {Learning both weights and connections for efficient neural networks},
year = {2015},
publisher = {MIT Press},
address = {Cambridge, MA, USA},
booktitle = {Proceedings of the 29th International Conference on Neural Information Processing Systems - Volume 1},
pages = {1135–1143},
numpages = {9},
location = {Montreal, Canada},
series = {NIPS'15}
}

@inproceedings{mrpc_dataset,
author = {Dolan, Bill and Brockett, Chris},
title = {Automatically Constructing a Corpus of Sentential Paraphrases},
booktitle = {Third International Workshop on Paraphrasing (IWP2005)},
year = {2005},
month = {January},
publisher = {Asia Federation of Natural Language Processing}}

@misc{rte_task,
      title={GLUE: A Multi-Task Benchmark and Analysis Platform for Natural Language Understanding}, 
      author={Alex Wang and Amanpreet Singh and Julian Michael and Felix Hill and Omer Levy and Samuel R. Bowman},
      year={2019},
      eprint={1804.07461},
      archivePrefix={arXiv},
      primaryClass={cs.CL},
      url={https://arxiv.org/abs/1804.07461}, 
}

@article{cola_dataset,
    author = {Warstadt, Alex and Singh, Amanpreet and Bowman, Samuel R.},
    title = {Neural Network Acceptability Judgments},
    journal = {Transactions of the Association for Computational Linguistics},
    volume = {7},
    pages = {625-641},
    year = {2019},
    month = {09},
    issn = {2307-387X},
    doi = {10.1162/tacl_a_00290},
    url = {https://doi.org/10.1162/tacl_a_00290},
    eprint = {https://direct.mit.edu/tacl/article-pdf/doi/10.1162/tacl_a_00290/1923083/tacl_a_00290.pdf},
}

@inproceedings{scienceQ,
    title={Learn to Explain: Multimodal Reasoning via Thought Chains for Science Question Answering},
    author={Lu, Pan and Mishra, Swaroop and Xia, Tony and Qiu, Liang and Chang, Kai-Wei and Zhu, Song-Chun and Tafjord, Oyvind and Clark, Peter and Ashwin Kalyan},
    booktitle={The 36th Conference on Neural Information Processing Systems (NeurIPS)},
    year={2022}
}

@inproceedings{commonsenseqa,
    title = "{C}ommonsense{QA}: A Question Answering Challenge Targeting Commonsense Knowledge",
    author = "Talmor, Alon  and
      Herzig, Jonathan  and
      Lourie, Nicholas  and
      Berant, Jonathan",
    editor = "Burstein, Jill  and
      Doran, Christy  and
      Solorio, Thamar",
    booktitle = "Proceedings of the 2019 Conference of the North {A}merican Chapter of the Association for Computational Linguistics: Human Language Technologies, Volume 1 (Long and Short Papers)",
    month = jun,
    year = "2019",
    address = "Minneapolis, Minnesota",
    publisher = "Association for Computational Linguistics",
    url = "https://aclanthology.org/N19-1421/",
    doi = "10.18653/v1/N19-1421",
    pages = "4149--4158"
}

@inproceedings{openbookQ,
    title = "Can a Suit of Armor Conduct Electricity? A New Dataset for Open Book Question Answering",
    author = "Mihaylov, Todor  and
      Clark, Peter  and
      Khot, Tushar  and
      Sabharwal, Ashish",
    editor = "Riloff, Ellen  and
      Chiang, David  and
      Hockenmaier, Julia  and
      Tsujii, Jun{'}ichi",
    booktitle = "Proceedings of the 2018 Conference on Empirical Methods in Natural Language Processing",
    month = oct # "-" # nov,
    year = "2018",
    address = "Brussels, Belgium",
    publisher = "Association for Computational Linguistics",
    url = "https://aclanthology.org/D18-1260/",
    doi = "10.18653/v1/D18-1260",
    pages = "2381--2391"
}

@inproceedings{boolq,
    title = "{B}ool{Q}: Exploring the Surprising Difficulty of Natural Yes/No Questions",
    author = "Clark, Christopher  and
      Lee, Kenton  and
      Chang, Ming-Wei  and
      Kwiatkowski, Tom  and
      Collins, Michael  and
      Toutanova, Kristina",
    editor = "Burstein, Jill  and
      Doran, Christy  and
      Solorio, Thamar",
    booktitle = "Proceedings of the 2019 Conference of the North {A}merican Chapter of the Association for Computational Linguistics: Human Language Technologies, Volume 1 (Long and Short Papers)",
    month = jun,
    year = "2019",
    address = "Minneapolis, Minnesota",
    publisher = "Association for Computational Linguistics",
    url = "https://aclanthology.org/N19-1300/",
    doi = "10.18653/v1/N19-1300",
    pages = "2924--2936"
}

@inproceedings{hellaswag,
    title={HellaSwag: Can a Machine Really Finish Your Sentence?},
    author={Zellers, Rowan and Holtzman, Ari and Bisk, Yonatan and Farhadi, Ali and Choi, Yejin},
    booktitle ={Proceedings of the 57th Annual Meeting of the Association for Computational Linguistics},
    year={2019}
}

@article{WinoGrande,
author = {Sakaguchi, Keisuke and Bras, Ronan Le and Bhagavatula, Chandra and Choi, Yejin},
title = {WinoGrande: an adversarial winograd schema challenge at scale},
year = {2021},
issue_date = {September 2021},
publisher = {Association for Computing Machinery},
address = {New York, NY, USA},
volume = {64},
number = {9},
issn = {0001-0782},
url = {https://doi.org/10.1145/3474381},
doi = {10.1145/3474381},
journal = {Commun. ACM},
month = aug,
pages = {99–106},
numpages = {8}
}

@misc{ARCchallenge,
      title={Think you have Solved Question Answering? Try ARC, the AI2 Reasoning Challenge}, 
      author={Peter Clark and Isaac Cowhey and Oren Etzioni and Tushar Khot and Ashish Sabharwal and Carissa Schoenick and Oyvind Tafjord},
      year={2018},
      eprint={1803.05457},
      archivePrefix={arXiv},
      primaryClass={cs.AI},
      url={https://arxiv.org/abs/1803.05457}, 
}

@article{datainf,
  title={Datainf: Efficiently estimating data influence in lora-tuned llms and diffusion models},
  author={Kwon, Yongchan and Wu, Eric and Wu, Kevin and Zou, James},
  journal={arXiv preprint arXiv:2310.00902},
  year={2023}
}

@misc{c4_dataset,
      title={Exploring the Limits of Transfer Learning with a Unified Text-to-Text Transformer}, 
      author={Colin Raffel and Noam Shazeer and Adam Roberts and Katherine Lee and Sharan Narang and Michael Matena and Yanqi Zhou and Wei Li and Peter J. Liu},
      year={2023},
      eprint={1910.10683},
      archivePrefix={arXiv},
      primaryClass={cs.LG},
      url={https://arxiv.org/abs/1910.10683}, 
}

@article{wanda,
  title={A Simple and Effective Pruning Approach for Large Language Models}, 
  author={Sun, Mingjie and Liu, Zhuang and Bair, Anna and Kolter, J. Zico},
  year={2023},
  journal={arXiv preprint arXiv:2306.11695}
}

@inproceedings{martens2015optimizing,
  title={Optimizing neural networks with kronecker-factored approximate curvature},
  author={Martens, James and Grosse, Roger},
  booktitle={International conference on machine learning},
  pages={2408--2417},
  year={2015},
  organization={PMLR}
}

@article{Rissanen1978Modeling,
  title={Modeling By Shortest Data Description*},
  author={Jorma Rissanen},
  journal={Autom.},
  year={1978},
  volume={14},
  pages={465-471},
  url={https://api.semanticscholar.org/CorpusID:30140639}
}

@inproceedings{lecun1989obd,
 author = {LeCun, Yann and Denker, John and Solla, Sara},
 booktitle = {Advances in Neural Information Processing Systems},
 editor = {D. Touretzky},
 pages = {},
 publisher = {Morgan-Kaufmann},
 title = {Optimal Brain Damage},
 url = {https://proceedings.neurips.cc/paper_files/paper/1989/file/6c9882bbac1c7093bd25041881277658-Paper.pdf},
 volume = {2},
 year = {1989}
}

@INPROCEEDINGS{hassibi1993obs,
	author={Hassibi, B. and Stork, D.G. and Wolff, G.J.},
	booktitle={IEEE International Conference on Neural Networks}, 
	title={Optimal Brain Surgeon and general network pruning}, 
	year={1993},
	volume={},
	number={},
	pages={293-299 vol.1},
	doi={10.1109/ICNN.1993.298572}}

@inproceedings{shazeer2017outrageously,
	title={ Outrageously Large Neural Networks: The Sparsely-Gated Mixture-of-Experts Layer},
	author={Noam Shazeer and *Azalia Mirhoseini and *Krzysztof Maziarz and Andy Davis and Quoc Le and Geoffrey Hinton and Jeff Dean},
	booktitle={International Conference on Learning Representations},
	year={2017},
	url={https://openreview.net/forum?id=B1ckMDqlg}
}

@article{fedus2022switch,
	title={Switch transformers: Scaling to trillion parameter models with simple and efficient sparsity},
	author={Fedus, William and Zoph, Barret and Shazeer, Noam},
	journal={Journal of Machine Learning Research},
	volume={23},
	number={120},
	pages={1--39},
	year={2022}
}

@inproceedings{koh2017understanding,
  title={Understanding black-box predictions via influence functions},
  author={Koh, Pang Wei and Liang, Percy},
  booktitle={International conference on machine learning},
  pages={1885--1894},
  year={2017},
  organization={PMLR}
}

@book{boyd2004convex,
  author    = {Boyd, Stephen and Vandenberghe, Lieven},
  title     = {Convex Optimization},
  publisher = {Cambridge University Press},
  year      = {2004},
  address   = {Cambridge, UK},
  note      = {Available at \url{https://web.stanford.edu/~boyd/cvxbook/}}
}

\newpage

\onecolumn

% \title{Curvature-Weighted Capacity Allocation: A Minimum Description Length Framework for Layer-Adaptive Large Language Model Optimization\\(Supplementary Material)}
\maketitle

\section*{Supplementary Overview}
\addcontentsline{toc}{section}{Supplementary Overview}
This supplement follows the same conceptual progression as the main paper.
Section~\ref{sec:related_work} places the proposed curvature-weighted framework
in the context of layer-quality estimation, second-order compression, MDL, and
adaptive capacity. Section~\ref{app:proofs} then supplies the detailed proofs
and the surrogate-regularization analysis underlying the main theoretical
claims. Section~\ref{sec:transfer} develops the transfer-stability result and
its cross-task diagnostic. The remaining sections provide implementation
details, complete experimental results, ablations, proxy validation, statistical
analysis, limitations, future work and the LLM-use disclosure.

\clearpage
\onecolumn
\setcounter{tocdepth}{2}
\tableofcontents
\clearpage
% \twocolumn

\appendix

\section{Related Work}
\label{sec:related_work}

The main paper combines a layer-wise curvature score with globally constrained
allocation and pruning programs. We first position these two ingredients
relative to prior work, beginning with layer-quality estimation and then moving
to second-order compression, MDL, and adaptive-capacity models.

\subsection{Layer Quality Estimation}
Quantifying the per-layer contribution to model performance is central
to pruning, expert allocation in mixture-of-experts (MoE) models,
and network compression.
\citet{layerIF} proposed LayerIF, which uses influence functions to
estimate layer quality and applies it to predict expert counts in
Mistral-7B and to determine layer-wise pruning ratios.
\citet{AVSS} introduced the Activation Variance-Sparsity Score (AVSS),
a combined measure of normalized activation variance and sparsity
that quantifies each layer's contribution to model performance.
\citet{layerImportanceEstimation} developed an accuracy-aware criterion
for ranking layer importance, used to guide quantization decisions.
A common limitation shared by these approaches is the absence of
curvature information: gradient magnitudes and activation statistics
do not account for the local geometry of the loss landscape.
Our gain $\zeta_k^2$ addresses this gap directly by incorporating
the inverse Hessian block, yielding a measure of \emph{reducible risk}
rather than raw gradient magnitude.

\subsection{Second-Order Methods for Model Compression}
Second-order information has been used for pruning since the seminal
Optimal Brain Damage \citet{lecun1989obd} and Optimal Brain Surgeon frameworks \citet{hassibi1993obs}, which identify parameters to remove
by computing the Hessian of the training loss.
More recent work scales these ideas to modern architectures using
diagonal Fisher approximations \citet{martens2015optimizing},
Kronecker-factored curvature \citet{botev2017practical}, and
randomized sketches.
Our work extends this tradition from individual weight pruning to
\emph{layer-level} capacity allocation, and introduces a convex
program that jointly optimizes the allocation across all layers under
a global budget.

\subsection{Generalization Bounds and Minimum Description Length}
Generalization theory asks how well a model trained on finite data
performs on unseen examples.
\citet{Valiant1984ATO} formalized PAC learning, establishing that
generalization depends on hypothesis space complexity and sample size.
\citet{ShalevShwartz2014UnderstandingML} extend this to infinite
hypothesis spaces via VC dimension.
For LLMs, classical bounds are vacuous due to the enormous number of
parameters; \citet{Lotfi2023NonVacuousGB} address this by introducing
compression-based bounds via SubLoRA, a low-dimensional nonlinear
parametrization that yields non-vacuous guarantees.
Our MDL objective is directly motivated by this line of work:
minimizing description length simultaneously controls generalization
and penalizes unnecessary model complexity, grounding our convex
programs in information-theoretic principles.

\subsection{Mixture-of-Experts and Adaptive Capacity}
Sparse MoE models \citet{shazeer2017outrageously, fedus2022switch}
increase model capacity without proportional compute cost by routing
tokens to a subset of expert layers.
Existing routing mechanisms are learned end-to-end and do not
explicitly account for the curvature-adjusted utility of adding
capacity at a given layer.
Our allocation program provides a principled, optimization-based
alternative that is complementary to learned routing: given a fixed
routing mechanism, it determines \emph{how much} capacity to assign
to each layer based on reducible risk.

Together, these lines of work motivate the two-part construction analyzed next:
a layer-wise second-order score that estimates locally reducible risk, and a
convex resource-allocation rule that converts those scores into globally
feasible decisions.

\section{Mathematical Foundations and Proofs}
\label{app:proofs}
With the surrounding literature established, we now give the detailed
mathematical arguments omitted from the main paper. The presentation follows
the logical dependency of the framework: first the layer-restricted quadratic
optimum, then the allocation and pruning programs, and finally the effect of
surrogate regularization.

\subsection{Layer-Restricted Quadratic Optimum}
\label{proof:lemma1}
We begin with Lemma~\ref{lem:layeropt}, because its regularized Newton step
defines the layer gain used by every subsequent optimization program.
Recall that $\widetilde Q_k(d_k)$ is the layer-restricted quadratic surrogate for
the change in empirical loss produced by a perturbation $d_k$ to layer $k$.
\begin{proof}
  The gradient of $\widetilde{Q}_k$ with respect to $d_k$ is
  $\nabla_{d_k}\widetilde{Q}_k = g_k + \widetilde{H}_{kk}\,d_k$.
  Setting this to zero and invoking $\widetilde{H}_{kk} \succ 0$ yields the unique
  stationary point $d_k^{\star} = -\widetilde{H}_{kk}^{-1}g_k$.
  Substituting back,
  \begin{align*}
    \widetilde{Q}_k(d_k^{\star})
    &= g_k^{\top}(-\widetilde{H}_{kk}^{-1}g_k)
      + \frac{1}{2}(-\widetilde{H}_{kk}^{-1}g_k)^{\top}
        \widetilde{H}_{kk}
        (-\widetilde{H}_{kk}^{-1}g_k) \\
    &= -\,g_k^{\top}\widetilde{H}_{kk}^{-1}g_k
      + \frac{1}{2}\,g_k^{\top}\widetilde{H}_{kk}^{-1}g_k
    \;=\;
    -\frac{1}{2}\,g_k^{\top}\widetilde{H}_{kk}^{-1}g_k. \qedhere
  \end{align*}
\end{proof}

\subsection{Closed-Form Allocation and Pruning Solutions}
\label{appendix:mdl_programs}
Having established the layer-wise gain, we next show how the two global
resource-allocation programs inherit convexity and reduce to one-dimensional
dual searches.

\subsubsection{Allocation: Convexity and Closed Form}
We first prove Theorem~\ref{thm:alloc}, which characterizes the optimal
continuous capacity allocation $e_k^\star$ under the global budget.
\begin{proof}
\label{proof:theorem2}
Each term $\alpha c_k e_k$ is linear.
Since $\phi$ is concave, $-\phi$ is convex, and since
$\gamma q_k^{\beta} \ge 0$, the term $-\gamma q_k^{\beta}\phi(e_k)$
is convex on $e_k \ge 0$.
The sum of convex functions over a convex feasible set is convex.
When $\phi$ is strictly concave and $q_k > 0$,
each term $-\gamma q_k^{\beta}\phi(e_k)$ is strictly convex,
giving strict convexity of the full objective.

Slater's condition \citet{boyd2004convex} holds (any strictly feasible point satisfies the
constraint with strict inequality), so strong duality applies.
The Lagrangian is
\begin{equation}
  \mathcal{L}(e,\lambda)
  =
  \sum_k
  \Bigl[\alpha c_k e_k - \gamma q_k^{\beta}\phi(e_k)\Bigr]
  + \lambda\!\left(\sum_k c_k e_k - B\right)
  ,\hspace{5em} \lambda \ge 0.    
\end{equation}

Stationarity at an interior point $e_k > 0$ gives
\[
  (\alpha + \lambda)\,c_k
  \;-\;
  \gamma q_k^{\beta}\,\phi'(e_k^{\star})
  \;=\; 0.
\]
For $\phi(e) = \log(1+e)$, $\phi'(e) = 1/(1+e)$, so
\[
  (\alpha + \lambda^{\star})\,c_k
  \;=\;
  \frac{\gamma q_k^{\beta}}{1 + e_k^{\star}},
\]
which yields Eq.~\eqref{eq:alloc_sol} after incorporating non-negativity.
Since $e_k(\lambda)$ is strictly decreasing and continuous in $\lambda$,
the sum $\sum_k c_k e_k(\lambda)$ is strictly decreasing, so
$\lambda^{\star}$ is unique and can be computed in
$O(K\log(1/\varepsilon))$ time by bisection.
Strict monotonicity of $e_k^{\star}$ in $q_k$ follows directly
from Eq.~\eqref{eq:alloc_sol}.
\end{proof}

\subsubsection{Pruning: Strong Convexity and Closed Form}
The pruning program is complementary to allocation: instead of distributing
additional capacity, it assigns layer-wise sparsity while enforcing a global
target. We now prove the closed-form characterization in
Eq.~\eqref{eq:prune_closed}.
\begin{proof}
\label{proof:prune}
The term $-bn_k\rho_k$ is linear.
If $\psi$ is strongly convex, $\eta q_k^{\kappa}\psi(\rho_k)$ is strongly
convex for $q_k > 0$, making the full objective strongly convex.
The feasible set $[0,1]^K \cap \{\sum_k n_k\rho_k \ge S\}$ is convex and
compact, so a unique minimizer exists.

The constraint $\sum_k n_k\rho_k \ge S$ is a lower-bound constraint,
so the Lagrangian is formed by subtracting the dual term:
\begin{equation}
\begin{array}{rl}
      \mathcal{L}(\rho,\lambda)=&
  \sum_k
  \Bigl[b\,n_k(1-\rho_k) + \eta q_k^{\kappa}\psi(\rho_k)\Bigr]
  - \lambda\!\Bigl(\sum_k n_k\rho_k - S\Bigr) - \sum_{k=1}^{K} \mu_k \rho_k + \sum_{k=1}^{K} \nu_k (\rho_k - \rho_{\max}),  \\\\
    &\lambda\ge 0, \,\,\, \mu_k \ge 0, \,\,\, \nu_k \ge 0
\end{array}
\label{eq:prune_dual}
\end{equation}

Stationarity at an interior point $\rho_k \in (0,\rho_{\max})$ gives

\[
  -b\,n_k
  + \eta\,q_k^{\kappa}\,\psi'(\rho_k^{\star})
  - \lambda^{\star}\,n_k - \mu_k + \nu_k
  \;=\; 0.
\]
For $\psi(\rho) = \rho^{2}$, $\psi'(\rho) = 2\rho$, so
\[
  \rho_k^{\star}
  \;=\;
  \frac{(b + \lambda^{\star})\,n_k}{2\,\eta\,q_k^{\kappa}},
\]
and box constraints $[0,\rho_{\max}]$ induce clipping, giving Eq.~\eqref{eq:prune_closed}.
Since $\lambda^{\star} \ge 0$, the numerator $b + \lambda^{\star} > 0$,
so $\rho_k^{\star} > 0$ on the interior.

The mapping $\lambda \mapsto \sum_k n_k\rho_k(\lambda)$ is strictly
increasing (larger $\lambda$ increases every $\rho_k$ through the
numerator $b + \lambda$), so $\lambda^{\star}$ is unique when binding
and can be found by bisection.
Strict decrease of $\rho_k^{\star}$ in $q_k$ follows directly
from Eq.~\eqref{eq:prune_closed}.
\end{proof}

\subsection{Validity Under Surrogate Regularization}
\label{sec:bias}

The preceding proofs establish optimality for the regularized quadratic
surrogate. We now connect that surrogate back to the actual empirical
objective and state the additional separation condition needed for layer
rankings to be certified.

The step $\Delta^{\star} = E_k^{\top}d_k^{\star}$ is derived from the surrogate
$\widetilde{H}_{kk}$, not from $H_{kk}$ directly.
We now show that the resulting ranking of layers by $\zeta_k^2$ is consistent
with the true objective decrease, and characterize the approximation gap explicitly.

Applying Eq.~\eqref{eq:taylor} to $\Delta^{\star}$ and using
$\Delta^{\star\top} H_{kk}\,\Delta^{\star}
 = \Delta^{\star\top}\widetilde{H}_{kk}\Delta^{\star} - \tau\|\Delta^{\star}\|^{2}$
gives
\begin{equation}
  \label{eq:truedec}
  L(\theta + \Delta^{\star}) - L(\theta)
  \;=\;
  -\frac{1}{2}\,\zeta_k^{2}
  \;-\;
  \frac{\tau}{2}\,\|\Delta^{\star}\|^{2}
  \;+\;
  R(\Delta^{\star}).
\end{equation}

\noindent
We bound the bias term as follows.
Since $d_k^{\star} = -\widetilde{H}_{kk}^{-1}g_k$ and
$\|E_k^{\top}d\| = \|d\|$ for any $d$,
\[
  \|\Delta^{\star}\|^{2}
  \;=\;
  \|d_k^{\star}\|^{2}
  \;=\;
  g_k^{\top}\widetilde{H}_{kk}^{-2}\,g_k.
\]
The spectral inequality $\widetilde{H}_{kk}^{-2} \preceq \lambda_{\min}(\widetilde{H}_{kk})^{-1}\widetilde{H}_{kk}^{-1}$
(which follows because all eigenvalues of $\widetilde{H}_{kk}^{-1}$ lie in
$(0,\lambda_{\min}(\widetilde H_{kk})^{-1}]$) then yields
\begin{equation}
  \label{eq:biasbound}
  \frac{\tau}{2}\,\|\Delta^{\star}\|^{2}
  \;\le\;
  \frac{\tau}{2\,\lambda_{\min}(\widetilde{H}_{kk})}\,\zeta_k^{2}.
\end{equation}
Combining Eq.~\eqref{eq:truedec} and Eq.~\eqref{eq:biasbound},
\[
  L(\theta + \Delta^{\star}) - L(\theta)
  \;\ge\;
  -\frac{1}{2}\left(1 + \frac{\tau}{\lambda_{\min}(\widetilde{H}_{kk})}\right)\zeta_k^{2}
  + R(\Delta^{\star}).
\]

The factor
$\tau/\lambda_{\min}(\widetilde H_{kk})$ measures the relative damping term
in this bound. When $H_{kk}\succeq0$, one has
$\lambda_{\min}(\widetilde H_{kk})\ge\tau$ and the factor is at most one.
For an indefinite block,
$\lambda_{\min}(\widetilde H_{kk})=\lambda_{\min}(H_{kk})+\tau$ may be
smaller than $\tau$; positive definiteness requires
$\tau>-\lambda_{\min}(H_{kk})$, and the ratio need not be at most one.

Small damping error alone does not preserve rankings when scores are nearly
tied. Define
\[
\varepsilon_k
:=
\frac{\tau}{2}\|d_k^\star\|_2^2
+
\frac{M}{6}\|d_k^\star\|_2^3.
\]
The ordering of layers $k$ and $\ell$ is certified whenever
\[
\frac12|\zeta_k^2-\zeta_\ell^2|>\varepsilon_k+\varepsilon_\ell.
\]
Neither batch normalization nor weight decay alone guarantees that an
empirical Hessian block is positive semidefinite.

This completes the within-task justification of the curvature score. We next
study a different source of error: replacing the target-task score vector by a
score vector estimated on a related source task.

\section{Transfer Stability Under Score Drift}
\label{sec:transfer}

The allocation and pruning programs depend on a normalized score
vector. We quantify the target-objective excess cost incurred when a source
score vector $q^{(A)}$ is substituted for the target score vector $q^{(B)}$,
while keeping the optimization objective, hyperparameters, and feasible set
fixed.

\subsection{Problem Setup}
Let $\mathcal{J}_B(x; q)$ denote the MDL program objective on task $B$,
as a function of the decision variable $x$ and quality inputs $q$.
Here $x = e \in \mathbb{R}^K_{\ge 0}$ for the allocation program
(Section~\ref{sec:alloc}) and $x = \rho \in [0,\rho_{max}]^K$ for the pruning
program (Section~\ref{sec:mdl_programs}),
with
\[
\Delta_K:=\{q\in\mathbb R_+^K:\mathbf 1^\top q=1\},
\qquad q\in\Delta_K.
\]
The score vector has $K$ coordinates and affine dimension $K-1$; it is
distinct from the decision variable $x$.

Let $\mathcal{X}$ denote the corresponding feasible set
(budget constraint or sparsity constraint).
Since Theorems~\ref{thm:alloc} and~\ref{thm:prune} guarantee
a unique minimizer for each admissible $q$, define
\[
  \widehat{x}_B(q)
  \;:=\;
  \arg\min_{x \in \mathcal{X}}\,\mathcal{J}_B(x;\, q).
\]
The \emph{transfer regret} of deploying the source-derived decision
$\widehat{x}_B(q^{(A)})$ on the target task is
\begin{equation}
  \label{eq:regret}
  \mathrm{Regret}_{A \to B}
  \;:=\;
  \mathcal{J}_B\!\left(\widehat{x}_B(q^{(A)});\, q^{(B)}\right)\\
  -
  \mathcal{J}_B\!\left(\widehat{x}_B(q^{(B)});\, q^{(B)}\right).
\end{equation}

Both terms are evaluated under the target objective and target scores
$q^{(B)}$, so Eq.~\eqref{eq:regret} is the excess cost of a misspecified
score input. The notation $\widehat x_B(q^{(A)})$ keeps
the target objective and feasible set fixed. Using
$\widehat x_A(q^{(A)})$ would additionally introduce objective,
hyperparameter, or feasible-set drift and would require separate error terms.

\subsection{Regularity Conditions}
Let $\Delta_K(q_{\min})=
\{q\in\Delta_K:q_k\ge q_{\min}\text{ for all }k\}$ for some
$q_{\min}>0$. We assume:
\begin{enumerate}[label=(\roman*)]
  \item $\mathcal{J}_B(\cdot;q)$ is $\sigma$-strongly convex on the common
  closed convex feasible set $\mathcal X$ for every
  $q\in\Delta_K(q_{\min})$;
  \item the decision gradient is Lipschitz in the score vector:
  \[
  \|\nabla_x\mathcal J_B(x;q)-\nabla_x\mathcal J_B(x;q')\|_2
  \le L_{xq}\|q-q'\|_2
  \]
  for all $x\in\mathcal X$ and
  $q,q'\in\Delta_K(q_{\min})$.
\end{enumerate}
For allocation with $\phi(e)=\log(1+e)$, one may take
\[
L_{xq}^{\mathrm{alloc}}
\le \gamma\beta\max_{t\in[q_{\min},1]}t^{\beta-1}.
\]
For pruning with $\psi(\rho)=\rho^2$,
\[
L_{xq}^{\mathrm{prune}}
\le 2\eta\kappa\rho_{\max}
\max_{t\in[q_{\min},1]}t^{\kappa-1}.
\]
The positive lower bound on $q_k$ is necessary for uniform constants when
$0<\beta<1$ or $0<\kappa<1$, and it also prevents the strong-convexity
modulus from vanishing in the pruning program.

\subsection{Transfer-Regret Bound}
With the score domain and regularity conditions fixed, we can state the
target-objective excess-cost guarantee.

\begin{theorem}[Transfer regret under score drift]
\label{thm:transfer}
Under assumptions~(i)--(ii),
\begin{equation}
\label{eq:regret_bound}
0\le \mathrm{Regret}_{A\to B}
\le
\frac{L_{xq}^2}{2\sigma}
\|q^{(A)}-q^{(B)}\|_2^2.
\end{equation}
If $|q_k^{(A)}-q_k^{(B)}|\le\delta_k$ for every $k$, then
\[
\mathrm{Regret}_{A\to B}
\le
\frac{L_{xq}^2}{2\sigma}\sum_{k=1}^K\delta_k^2.
\]
\end{theorem}

\begin{proof}
Write $x_A=\widehat x_B(q^{(A)})$ and
$x_B=\widehat x_B(q^{(B)})$. Strong convexity of
$\mathcal J_B(\cdot;q^{(B)})$ gives
\begin{equation*}
\mathcal J_B(x_A;q^{(B)})-\mathcal J_B(x_B;q^{(B)})
\le
\left\langle\nabla_x\mathcal J_B(x_A;q^{(B)}),x_A-x_B\right\rangle
-\frac{\sigma}{2}\|x_A-x_B\|_2^2.
\end{equation*}
Because $x_A$ minimizes $\mathcal J_B(\cdot;q^{(A)})$ over the closed
convex set $\mathcal X$, its variational inequality is
\[
\left\langle\nabla_x\mathcal J_B(x_A;q^{(A)}),x_B-x_A\right\rangle\ge0.
\]
Equivalently,
$\langle\nabla_x\mathcal J_B(x_A;q^{(A)}),x_A-x_B\rangle\le0$.
Subtracting this nonpositive term and applying assumption~(ii) yields
\begin{equation*}
\mathrm{Regret}_{A\to B}
\le
L_{xq}\|q^{(A)}-q^{(B)}\|_2\|x_A-x_B\|_2
-\frac{\sigma}{2}\|x_A-x_B\|_2^2.
\end{equation*}
The right-hand side is maximized over
$r=\|x_A-x_B\|_2\ge0$ at
$r=L_{xq}\|q^{(A)}-q^{(B)}\|_2/\sigma$, which gives
Eq.~\eqref{eq:regret_bound}. Nonnegativity follows from the optimality of
$x_B$ under the target score vector. The coordinate-wise statement follows
from $\|q^{(A)}-q^{(B)}\|_2^2\le\sum_k\delta_k^2$.
\end{proof}

\subsection{Interpretation and Boundary Validity}
The bound is valid for interior and boundary solutions,
including active budget, sparsity, and box constraints; no multiplier-sensitivity
argument or additional smoothness constant is required. The factor
$L_{xq}^2/(2\sigma)$ separates sensitivity to score perturbations from the
strong-convexity conditioning of the target surrogate program.

\subsection{Cross-Task Transfer Diagnostic}
We conduct a cross-task diagnostic on Gemma-7B by deriving
the expert allocation from CommonsenseQA and deploying that allocation on the
remaining tasks. Because the experimental weights are LayerIF-derived proxies,
we report the squared proxy-score drift
\[
\|\widehat q^{(\mathrm{source})}-\widehat q^{(\mathrm{target})}\|_2^2,
\]
with each proxy vector normalized to sum to one. The accompanying performance
quantity is a downstream accuracy difference, not the nonnegative optimization
regret in Eq.~\eqref{eq:regret}.
    
    \begin{table}[tbh]
    \caption{Cross-task proxy-score drift and downstream accuracy difference for a CommonsenseQA-derived allocation.}
    % \begin{minipage}{0.38\textwidth}    
    \centering
    % \begin{tabular}{llc}
    \begin{tabular}{l >{\raggedright\arraybackslash}p{1.7cm} >{\centering\arraybackslash}p{1.7cm}}
    \toprule
    \textbf{Dataset} & \textbf{Proxy-score drift} & \textbf{Accuracy difference} \\ \hline
    CommonQA (Source) & 0.0000 &  0.00 \\ \midrule
    MRPC              & 0.2447 & -0.34 \\ \midrule
    CoLA              & 0.6273 &  1.82 \\ \midrule
    OpenBookQA        & 0.3127 &  1.20 \\ \midrule
    ScienceQA         & 0.6698 &  1.44 \\ 
    \bottomrule
    \end{tabular}    
    \label{tab:drift_score}
    \end{table}
    \vspace{1em}
The negative MRPC entry confirms that the third column is not
the theorem's optimization regret, which is nonnegative by definition. The table
therefore provides a qualitative cross-task diagnostic. The direct numerical
quantity associated with Theorem~\ref{thm:transfer} is
\[
R_J=
\mathcal J_B(\widehat x_B(q^{(A)});q^{(B)})
-
\mathcal J_B(\widehat x_B(q^{(B)});q^{(B)})\ge0,
\]
with upper bound
$L_{xq}^2\|q^{(A)}-q^{(B)}\|_2^2/(2\sigma)$.

The theorem controls the surrogate optimization objective, whereas the
cross-task experiment additionally reflects training noise, integer expert
counts, and downstream evaluation. We therefore keep the two quantities
separate and now turn to the implementation details and complete empirical
results.

\section{Experimental Details and Full Results}
\label{appendix:experiments}
\subsection{Hardware}
\label{appendix:hardware}
Mistral-7B experiments are run on 4$\times$ NVIDIA L40S GPUs;
Gemma-7B experiments are run on 4$\times$ NVIDIA A6000 GPUs.

\subsection{Hyperparameters}
\label{appendix:hyperparameter}
Table~\ref{tab:hyperparams} summarizes all hyperparameter settings.
For the pruning program (Algorithm~\ref{alg:prune}):
$b = 16$ bits per retained parameter, $\eta = 2$, $\kappa = 1$.
For the allocation program (Algorithm~\ref{alg:alloc}):
$\alpha = 0.5$, $\gamma = 0.9$, $\beta = 1$.
Layer sizes $n_k$ are set to the number of parameters in layer $k$,
and per-unit costs $c_k$ are set to the FLOPs of a single LoRA
expert at layer $k$.

\begin{table}[h]
\centering
\caption{Hyperparameter configurations for allocation and pruning.}
\label{tab:hyperparams}
\begin{tabular}{llcc}
\toprule
\textbf{Program} & \textbf{Parameter} & \textbf{Mistral-7B} & \textbf{Gemma-7B} \\
\midrule
\multirow{3}{*}{Allocation}
  & $\alpha$   & 0.5    & 0.5  \\
  & $\gamma$   & 0.9    & 0.9  \\
  & budget scaling, $\sigma$ & 0.0276 & 0.02 \\
\midrule
\multirow{3}{*}{Pruning}
  & $b$ (bits) & 16     & 16   \\
  & $\eta$     & 2      & 2    \\
  & $\kappa$   & 1      & 1    \\
  & Sparsity target, $S$ & 50\% & 50\% \\
\bottomrule
\end{tabular}
\end{table}

\subsection{Complete Expert-Allocation Results}
Tables~\ref{tab:mistral_expert_alloc} and~\ref{tab:exp_gemma} provide the
dataset-level results summarized in Section~\ref{subsec:results} of the main
paper. These complete tables separate the effect of the convex decision rule
from the influence-style score estimator shared with LayerIF.

\begin{table}[tbh]
\centering
\caption{Expert allocation accuracy (\%) on Mistral-7B-v0.1 (5 epochs).
         Best average per category in \textbf{bold}.}
\label{tab:mistral_expert_alloc}
\begin{tabular}{lrr|rr}
\toprule
& \multicolumn{2}{c|}{\textbf{All}} & \multicolumn{2}{c}{\textbf{+ve}} \\
\textbf{Dataset} & \textbf{LayerIF} & \textbf{MDL} & \textbf{LayerIF} & \textbf{MDL} \\
\midrule
CoLA        & 85.62 & 87.44 & 86.10 & 85.90 \\
MRPC        & 83.59 & 82.72 & 82.84 & 84.23 \\
CommonsenseQA & 81.24 & 80.43 & 81.40 & 80.18 \\
OpenBookQA  & 85.20 & 85.00 & 85.80 & 86.40 \\
ScienceQA   & 66.41 & 79.77 & 80.80 & 83.59 \\
\midrule
Average     & 80.41 & \textbf{83.07} & 83.39 & \textbf{84.06} \\
\bottomrule
\end{tabular}
\end{table}

\begin{table}[tbh]
\centering
\caption{Expert allocation accuracy (\%) on Gemma-7B, +ve variant only (5 epochs).
         Best average in \textbf{bold}.}
\label{tab:exp_gemma}
\begin{tabular}{lrr}
\toprule
\textbf{Dataset} & \textbf{LayerIF} & \textbf{MDL (+ve)} \\
\midrule
CoLA          & 87.15 & 87.34 \\
MRPC          & 84.41 & 84.99 \\
CommonsenseQA & 82.64 & 81.57 \\
OpenBookQA    & 88.40 & 88.80 \\
ScienceQA     & 94.69 & 94.92 \\
\midrule
Average       & 87.46 & \textbf{87.52} \\
\bottomrule
\end{tabular}
\end{table}

The full results establish the primary comparison. We next examine which parts
of the framework drive those outcomes by varying the allocation and pruning
hyperparameters and by comparing alternative score-to-decision pipelines.

\section{Ablations and Proxy Validation}
This section complements the main experiments in three steps. We first compare
the proposed allocation rule with uniform, hand-designed, and heuristic
alternatives. We then study sensitivity to the exponents $\beta$ and $\kappa$.
Finally, we compare direct evaluations of $\zeta_k^2$ with scalable proxy
scores on a smaller network where the regularized curvature quantity can be
computed explicitly.

\subsection{Mistral-7B Expert-Allocation Comparison}
We compare uniform and non-uniform MoLA allocations
\citep{MolA}, AlphaLoRA \citep{alphalora}, a gradient-norm/Fisher score
combined with the LayerIF heuristic, the original LayerIF decision rule
\citep{layerIF}, and our convex allocation program using the same scalable
influence-style score inputs. All methods are evaluated on Mistral-7B-v0.1
under the common five-epoch protocol.

\begin{table}[tbh]
  \centering
  \caption{Accuracy (\%) across expert-allocation strategies (Mistral-7B-v0.1, 5 epochs, All).}
  \label{tab:comparison}
  \begin{tabular}{lccccccc}
    \toprule
    Dataset & Uniform & \multicolumn{2}{c}{Non-uniform} & AlphaLora & Fisher & LayerIF & MDL \\
     & (MoLA 5555) & MoLA 2468 & MoLA 8642 & & (grad.\ norm) & & \\
    \midrule
    Cola         & 85.52 & 86.29 & 85.43 & 87.44 & 85.33 & 85.62 & 87.44 \\
    MRPC         & 83.94 & 82.61 & 84.23 & 83.54 & 80.87 & 83.59 & 82.72 \\
    CommonsenseQA        & 81.24 & 81.08 & 63.47 & 81.90 & 79.52 & 81.24 & 80.43 \\
    OpenbookQA       & 85.40 & 86.60 & 82.20 & 82.00 & 86.20 & 85.20 & 85.00 \\
    ScienceQA & 81.21 & 76.39 & 82.55 & 72.26 & 78.33 & 66.41 & 79.77 \\
    \bottomrule
  \end{tabular}
\end{table}
Table~\ref{tab:comparison} shows that the proposed method remains
competitive with the strongest alternatives while providing an explicit
continuous objective and a globally enforced resource budget. The comparison
also clarifies that the theoretical contribution concerns the score-to-decision
rule rather than a uniformly superior score estimator.

\subsection{Sensitivity to $\beta$}
We next isolate the gain-emphasis exponent $\beta$ in
Algorithm~\ref{alg:alloc}. Using Gemma-7B, we vary $\beta$ while holding the
remaining allocation hyperparameters and training protocol fixed. This
experiment measures how strongly concentrating capacity on the largest scores
affects dataset-level and average accuracy.
\begin{figure}[tbh]
    
    \begin{minipage}{0.58\textwidth}
        \centering
        \subcaptionbox{}{\includegraphics[width=\columnwidth]{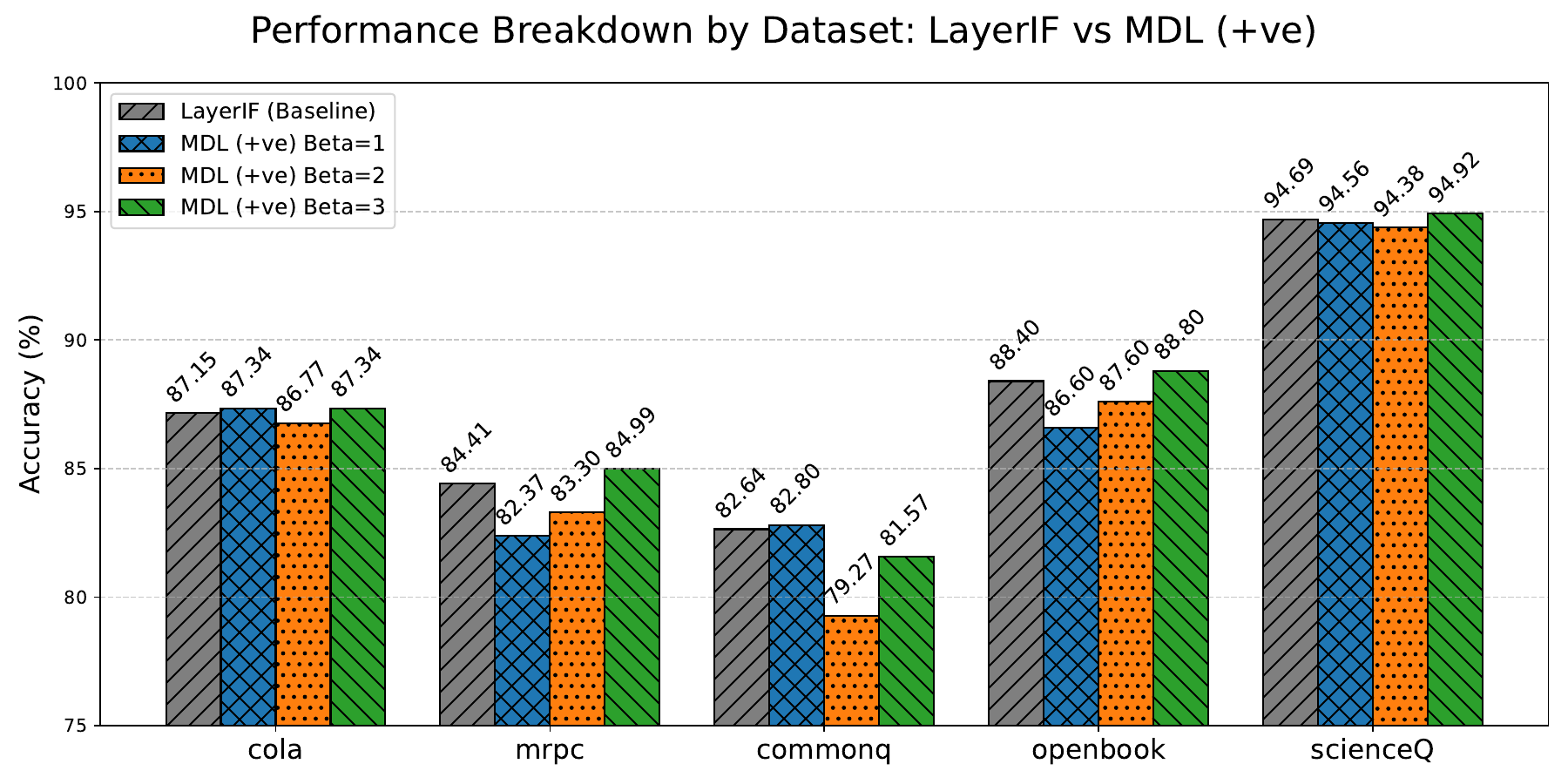}}
        \label{fig:beta_study}
    \end{minipage} 
    \hfill
    \begin{minipage}{0.48\textwidth}
        \subcaptionbox{}{\includegraphics[width=\columnwidth]{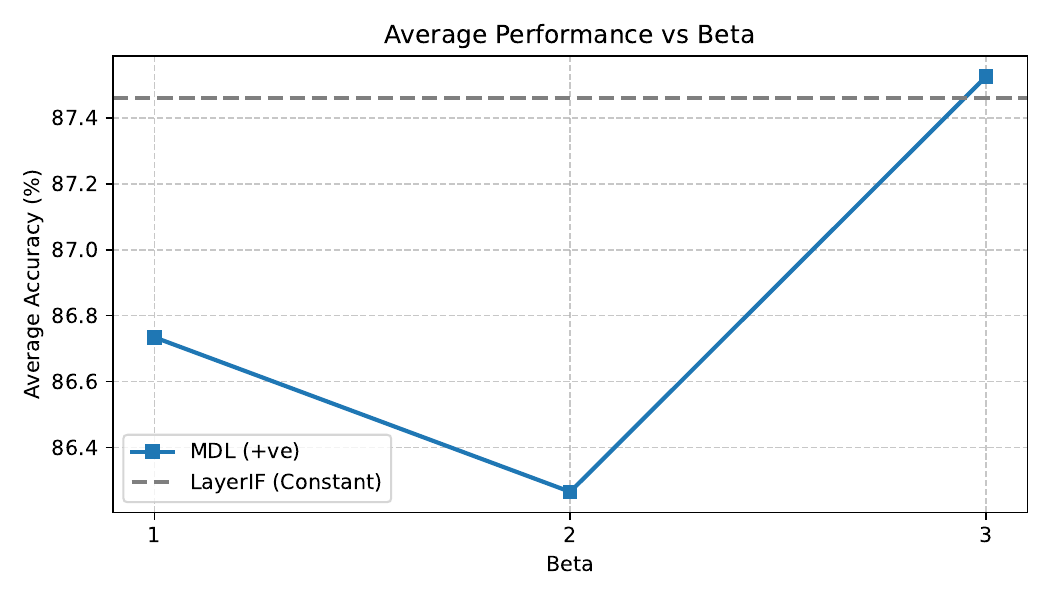}}
        \label{fig:beta_study_avg}
    \end{minipage}
    \caption{(a) reveals the detailed accuracy per dataset as Beta goes from $1-3$. (b) As shown for Gemma-7B, $\beta=3$ gives, on average, a better accuracy across the datasets.}
\end{figure}

\subsection{Sensitivity to $\kappa$}
The pruning exponent $\kappa$ controls how sharply the degradation penalty
distinguishes high- and low-score layers. We vary $\kappa$ on Gemma-7B while
holding the remaining pruning hyperparameters fixed, and report zero-shot
accuracy for Magnitude, Wanda, and SparseGPT pruning over the tasks listed in
Section~\ref{sec:experiment} of the main paper.

\begin{figure}[!tbh]
    \centering
    \subcaptionbox{}{\includegraphics[width=\linewidth]{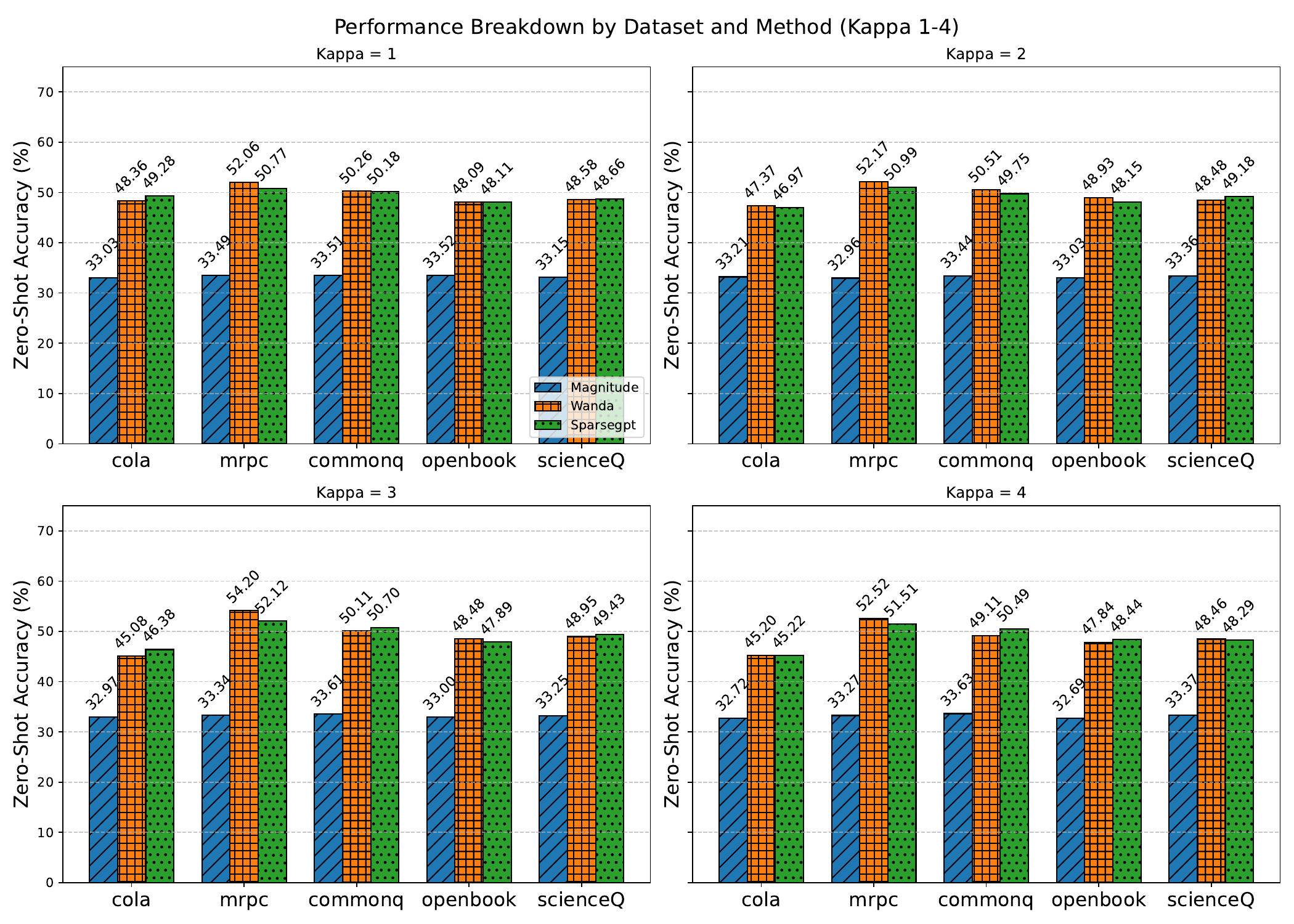}}

    \centering
    \subcaptionbox{}{\includegraphics[width=0.6\linewidth]{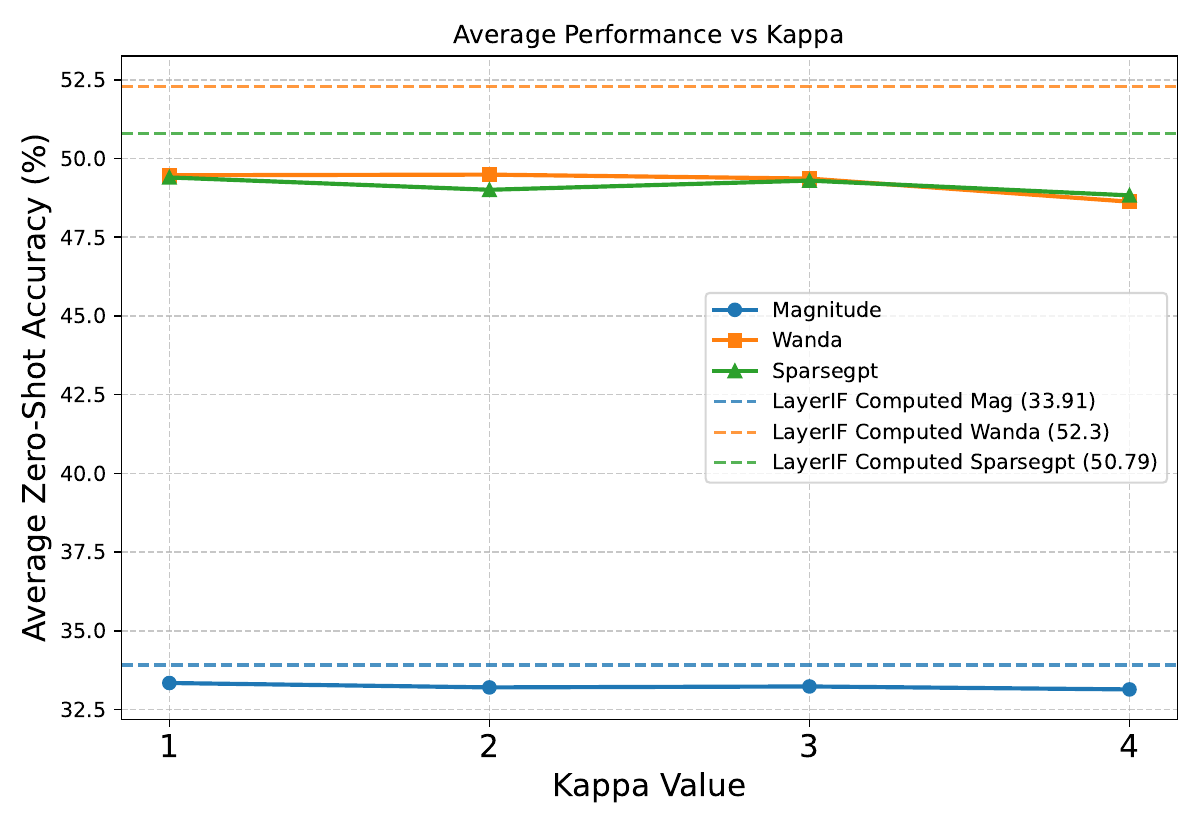}}
    \caption{(a) Performance breakdown for the three pruning types across varying values of $\kappa$ (b) Across the pruning types, an increase in $\kappa$ doesn't seem to be helpful when compared to the LayerIF baselines.   }
    \label{fig:kappa_studies}
\end{figure}

\subsection{Direct and Proxy Estimates of $\zeta^2$ on an 8-Layer MLP}
The preceding ablations vary the optimization programs while keeping the score
pipeline fixed. We now examine that pipeline directly on an 8-layer
multilayer perceptron, where
$\zeta_k^2(\tau)=g_k^\top\widetilde H_{kk}^{-1}g_k$ can be computed and
compared with scalable alternatives. The damping parameter $\tau$ is central
to this comparison: as $\tau I$ increasingly dominates $H_{kk}$ in
Eq.~\eqref{eq:surrogate},
\[
\widetilde H_{kk}^{-1}\approx \tau^{-1}I,
\qquad
\zeta_k^2(\tau)\approx \tau^{-1}\|g_k\|_2^2,
\]
so the regularized second-order score approaches a rescaled gradient-norm
ranking. Tables~\ref{tab:tau_0.001}--\ref{tab:tau_2} show how the proxy
correlations evolve across this regime.

\begin{table}[htbp]
    \centering
    
    \begin{minipage}{0.48\textwidth}
        \centering
        \caption{Tau / damping = 0.001}
        \label{tab:tau_0.001}
        \begin{tabular}{lcc}
        \toprule
        Proxy/Baseline & Spearman & Kendall \\
        \midrule
        LayerIF\_style & $0.659 \pm 0.224$ & $0.571 \pm 0.210$ \\
        diag\_fisher & $0.484 \pm 0.204$ & $0.381 \pm 0.187$ \\
        grad\_norm & $0.611 \pm 0.157$ & $0.452 \pm 0.135$ \\
        \bottomrule
        \end{tabular}
    \end{minipage}\hfill
    \begin{minipage}{0.48\textwidth}
        \centering
        \caption{Tau / damping = 0.01}
        \label{tab:tau_0.01}
        \begin{tabular}{lcc}
        \toprule
        Proxy/Baseline & Spearman & Kendall \\
        \midrule
        LayerIF\_style & $0.754 \pm 0.185$ & $0.643 \pm 0.210$ \\
        diag\_fisher & $0.349 \pm 0.238$ & $0.310 \pm 0.187$ \\
        grad\_norm & $0.706 \pm 0.137$ & $0.571 \pm 0.117$ \\
        \bottomrule
        \end{tabular}
    \end{minipage}
    
    \vspace{1.5em}%
    
    \begin{minipage}{0.48\textwidth}
        \centering
        \caption{Tau / damping = 0.1}
        \label{tab:tau_0.1}
        \begin{tabular}{lcc}
        \toprule
        Proxy/Baseline & Spearman & Kendall \\
        \midrule
        LayerIF\_style & $0.881 \pm 0.168$ & $0.833 \pm 0.236$ \\
        diag\_fisher & $0.214 \pm 0.418$ & $0.167 \pm 0.379$ \\
        grad\_norm & $0.690 \pm 0.085$ & $0.571 \pm 0.058$ \\
        \bottomrule
        \end{tabular}
    \end{minipage}\hfill
    \begin{minipage}{0.48\textwidth}
        \centering
        \caption{Tau / damping = 1}
        \label{tab:tau_1}
        \begin{tabular}{lcc}
        \toprule
        Proxy/Baseline & Spearman & Kendall \\
        \midrule
        LayerIF\_style & $0.738 \pm 0.178$ & $0.619 \pm 0.221$ \\
        diag\_fisher & $0.381 \pm 0.152$ & $0.333 \pm 0.121$ \\
        grad\_norm & $0.825 \pm 0.119$ & $0.738 \pm 0.121$ \\
        \bottomrule
        \end{tabular}
    \end{minipage}
    
    \vspace{1.5em}%
    
    \begin{minipage}{0.48\textwidth}
        \centering
        \caption{Tau / damping = 2}
        \label{tab:tau_2}
        \begin{tabular}{lcc}
        \toprule
        Proxy/Baseline & Spearman & Kendall \\
        \midrule
        LayerIF\_style & $0.595 \pm 0.185$ & $0.476 \pm 0.168$ \\
        diag\_fisher & $0.444 \pm 0.011$ & $0.381 \pm 0.034$ \\
        grad\_norm & $0.929 \pm 0.034$ & $0.833 \pm 0.067$ \\
        \bottomrule
        \end{tabular}
    \end{minipage}
    
\end{table}

The results reveal a damping-dependent transition. At small-to-moderate
damping, the LayerIF-style proxy is most strongly aligned with the direct
regularized score, whereas at larger damping the gradient norm becomes the
closest proxy, as predicted by
$\widetilde H_{kk}^{-1}\approx \tau^{-1}I$. This diagnostic does not establish
proxy equivalence for LLMs, but it clarifies when an influence-style proxy can
retain information beyond raw gradient magnitude.

Having examined both program hyperparameters and score proxies, we finally
assess whether the observed allocation differences are statistically
distinguishable under the current evaluation budget.

\section{Statistical Significance and Interpretation}
To assess the reliability of the allocation comparison, we performed paired
$t$-tests between the proposed MDL-inspired decision rule and LayerIF on the
matched model--dataset evaluations reported below. Because the two methods use
the same influence-style layer-score inputs, this analysis isolates the
difference between their allocation rules rather than the quality of the
underlying sensitivity estimator.

\begin{table}[htbp]

        \caption{Paired comparisons between LayerIF and the proposed allocation rule. None of the reported differences is significant at the 0.05 level.}
    \begin{minipage}{0.48\textwidth}
        \centering
        \label{tab:mistral-p-value-posIF}
        \begin{tabular}{lcc}
        \toprule
        Algorithm (Mistral) & Average ($\%$) & p-value \\
        \midrule
        LayerIF (+ve) & $83.39$ & $0.1906$ \\
        MDL (+ve) & $84.06$ & $-$ \\
        \midrule
        Mean diff & $0.67$ & $-$\\
        \bottomrule
        \end{tabular}
    \end{minipage}\hfill
    \begin{minipage}{0.48\textwidth}
        \centering
        \label{tab:mistral-p-value-allIF}
        \begin{tabular}{lcc}
        \toprule
        Algorithm (Mistral) & Average ($\%$) & p-value \\
        \midrule
        LayerIF (All) & $80.41$ & $0.1917$ \\
        MDL (All) & $83.07$ & $-$ \\
        \midrule
        Mean diff & $2.66$ & $-$ \\
        \bottomrule
        \end{tabular}
    \end{minipage}

    \vspace{1.3em}

    {\centering
    \begin{minipage}{0.48\textwidth}
        \centering
        \begin{tabular}{lcc}
        \toprule
        Algorithm (Gemma) & Average ($\%$) & p-value \\
        \midrule
        LayerIF (All) & $87.46$ & $0.4162$ \\
        MDL (All) & $87.52$ & $-$ \\
        \midrule
        Mean diff & $0.07$ & $-$ \\
        \bottomrule
        \end{tabular}
    \end{minipage}\par}
\end{table}

The differences are not statistically significant at the 0.05 level, so
the experiments support comparability rather than a universal accuracy
improvement. This outcome should be interpreted in the context of the design:
our method and LayerIF intentionally share the same influence-style layer
scores, and the comparison therefore isolates the decision rule. Under the
current evaluation budget, replacing the LayerIF knapsack heuristic with the
convex MDL-inspired program preserves competitive empirical performance while
adding an explicit risk--complexity objective, global budget feasibility, and
the transfer-stability guarantee proved above.

% \newpage
\section{Limitations} \label{sec:limit}
The quadratic degradation model $\psi(\rho)=\rho^2$ may
underestimate pruning sensitivity in architectures with highly
heterogeneous layer widths, as seen in the Gemma-7B Wanda results.
The Hessian surrogate $\widetilde{H}_{kk}$ requires a curvature
approximation (e.g., diagonal Fisher or K-FAC) whose quality affects
the accuracy of $\zeta_k^2$.
The experiments instantiate the programs with LayerIF-derived proxy
weights rather than directly computed Newton-decrement gains. They therefore
evaluate the convex decision rule, but do not directly establish the empirical quality of the proposed curvature estimator.

\section{Future work} \label{sec:future}
Natural extensions include: (i) richer degradation penalties
$\psi$ calibrated to specific pruning methods;
(ii) joint optimization of allocation and pruning in a single
program; and (iii) online updating of $\zeta_k^2$ during fine-tuning
to track curvature drift adaptively.

This concludes the theoretical and empirical supplement. We finish with the
required disclosure concerning the use of language-model tools during
preparation.

\section{LLM Use Disclosure}
We employed Large Language Models (LLMs) to refine the text for grammar and clarity. Additionally, LLMs were used for debugging the published codebases of \citet{layerIF}, \citet{datainf} and in writing  visualization code. We confirm that LLMs were not used to implement any of the core algorithms or methodologies proposed in this work.

\end{document}